\documentclass[12pt]{article}
\usepackage{inputenc}
\usepackage[margin = 1in]{geometry}
\usepackage{amsfonts}
\usepackage{amsmath}
\usepackage{amssymb}
\usepackage[english]{babel}
\usepackage{amsthm}
\usepackage{tikz}
\usepackage{float}
\usepackage[export]{adjustbox}
\usepackage{subcaption}
\usepackage{bbm}
\usepackage{comment}
\usepackage{natbib}
\usepackage{enumerate}
\usepackage{hyperref}
\usepackage{cleveref}
\usepackage{arydshln}
\usepackage{pgfplots}
\usepackage{pgfplots}
\usepackage{setspace}
\usepackage{graphicx}
\usepackage[export]{adjustbox}
\pgfplotsset{compat=1.18}
\usetikzlibrary{calc}
\usetikzlibrary{patterns}
\usetikzlibrary{arrows}
\setcitestyle{authoryear,open={(},close={)}}

\newcommand{\E}{ \mathbbm{E}}
\newcommand{\R}{ \mathbbm{R}}

\excludecomment{note}

\theoremstyle{definition}
\newtheorem{theorem}{Theorem}
\newtheorem{proposition}{Proposition}
\newtheorem{lemma}{Lemma}

\newtheorem{corollary}{Corollary}

\newtheorem{assumption}{Assumption}

\theoremstyle{remark}
\newtheorem{remark}{Remark}

\DeclareMathOperator*{\argmax}{\arg\max}

\newcommand{\und}{\underline}
\newcommand{\vst}{\vspace{3mm}}

\setdisplayskipstretch{0.9}

\title{Automation and Task Allocation Under Asymmetric Information}
\author{Quitz\'{e} Valenzuela-Stookey\thanks{Department of Economics, University of California, Berkeley. Email: quitze@berkeley.edu. I am grateful to Federico Echenique and Eliott Lipnowski; and seminar participants at the Cowles Summer Conference, SAET Summer meetings, and NYU; for discussion and feedback.}}
\date{November 4, 2025}

\begin{document}
\maketitle

\begin{abstract}
    A firm can complete the tasks needed to produce output using either machines or workers. Unlike machines, workers have private information about their preferences over tasks. I study how this information asymmetry shapes the mechanism used by the firm to allocate tasks across workers and machines. I identify important qualitative differences between the mechanisms used when information frictions are large versus small. When information frictions are small, tasks are substitutes: automating one task lowers the marginal cost of other tasks and reduces the surplus generated by workers. When frictions are large, tasks can become complements: automation can raise the marginal cost of other tasks and increase the surplus generated by workers. The results extend to a setting with multiple firms competing for workers. 
\end{abstract}

The question of how automation affects workers has occupied economists and policymakers since the early days of the industrial revolution.\footnote{See for example \citep{stockingframes1812}, and \cite{keynes1930economic}.} This debate has taken on renewed urgency following recent advances in artificial intelligence. Even conservative assessments conclude that this technology is likely to have a significant impact on labor markets in the near future \citep{acemoglu2025simple}.\footnote{At the other end of the spectrum, some industry figures predict that AI will advance to the point at which it ``outperforms humans at most economically valuable work'' within 20 years \citep{OpenAIOurStructure,  korinek2024scenarios}.}

A growing literature in economics studies the effects of technological change by modeling production as a task-allocation problem: firms produce consumption goods by combining the services of a range of tasks, each of which can be completed by a worker or a machine \citep{acemoglu2018race}.\footnote{The idea of modeling production as task allocation dates at least to \cite{roy1951some}, with important contributions by \cite{zeira1998workers} and \cite{acemoglu2001productivity} among others.} Importantly, there is no sharp distinction between capital and labor; automation can be driven by reductions in the cost of using machines for certain tasks, or by improvements in machine performance. Task-allocation models are thus well-suited to study the substitution of capital for labor in a world with an evolving and increasingly blurred boundary between humans and machines as inputs into production. 

A critical difference between human workers and machines, however, is that humans have preferences over the tasks they are assigned. Importantly, these preferences are a worker's private information and are not directly observable to their employer.\footnote{Machines have costs for completing tasks, which are an important feature of the model studied here, but machines cannot (for now) be said to have private information about these costs.} Nonetheless, workers' preferences may matter to the firm when assigning tasks; perhaps because the firm is competing with other firms to attract workers, or because the firm is concerned about burnout and staff revolt. Thus the firm faces a screening problem: it may be forced to distort the allocation of tasks, and to grant ``information rents'' in the form of more favorable allocations to some workers, in order to incentivize workers to reveal their preferences. 

This paper studies the interaction between asymmetric information and technological change in task-allocation problems. The implications of this interaction are not clear a priori. On one hand, distortions which raise the cost of using labor to complete tasks could amplify a firm's incentives to substitute towards machines. For example, \cite{acemoglu2024automation} suggest that factors that introduce a wedge between wages and workers' opportunity costs, such as minimum-wage regulations or bargaining power, could have this effect. On the other hand, workers' private information could constrain the firm's ability to displace them from certain tasks; it is well-understood that private information gives agents leverage in contracting situations \citep{stiglitz1975theory, hylland1979efficient, myerson1981optimal}.

The main economic message of this paper is that there are important qualitative differences between the effects of technological changes on the allocation of tasks to workers, depending on whether information asymmetries regarding workers' preferences are large or small. The basic intuition is the following. Under asymmetric information, allocating a particular task to a worker as part of the task-allocation mechanism plays a dual role for a firm. The direct effect is that by completing the task the worker generates some surplus for the firm. However there is also an \textit{incentive effect}, in that the assignment affects workers' incentives in the firm's screening problem. When the incentive effect is dominant, i.e., informational asymmetries are sufficiently large, I show that tasks can become complements for the firm: if assigning a worker to task $A$ becomes relatively more profitable, the firm may also increase how much of task $B$ they are given, so as to provide the correct incentives for the worker to reveal their preferences. In contrast, when informational asymmetries are small, tasks are always substitutes. Whether tasks are complements or substitutes has important implications for how the firm responds to technological changes which affect the value of assigning certain tasks to machines. 

To formalize this intuition, I study a simple task-allocation model in which a final consumption good is produced by combining two types of tasks, $A$ and $B$, and the output level is determined by the number of completed units of each task. In the baseline model I consider a market with a single firm (in \Cref{sec:competitive} I show that the conclusions extend to markets with multiple firms). The firm employs a continuum of workers, each of whom belongs to one of a discrete set of groups. A worker's group identity determines their productivity, measured by the marginal cost to the firm of assigning them to complete each of the two tasks. Group identity, and hence productivity, is known to the firm. A worker's private information concerns their preferences over tasks, represented by their subjective marginal cost of completing each of the two tasks. The firm knows only the distribution of workers' preferences (which may differ across groups). 

In the single-firm model, workers' preferences matter to the firm because it must respect a ``workload constraint'': the total (subjective) cost to a worker of completing their assigned tasks cannot exceed some reservation level (we allow for the possibility that this threshold depends on the worker's preferences).\footnote{With multiple firms (as in \Cref{sec:competitive}) workers' preferences also matter because firms compete with each other to attract workers.} The simplest interpretation of this constraint is that the firm wishes to avoid the burnout and turnover associated with over-worked workers.\footnote{An alternative interpretation of the within-firm workload constraint is that the threshold represents the workers' welfare under a status-quo task-allocation system, and the firm needs to get all workers on board with any new system. Such hold-harmless clauses are present in many labor-market settings \citep{dinerstein2021}.} 

The firm also has access to a machine that can be used to complete tasks. Like workers, the machine's performance on each of the two tasks is known to the firm. However rather than preferences, the machine has costs for completing the tasks, measured in terms of a resource that we call ``computation'', and these costs are also known to the firm.  

We would like to understand how the firm responds to technological changes which affect the machine, either in terms of its performance on certain tasks, or in terms of the computation resources required. To this end, I first characterize the firm's optimal \textit{labor-market mechanism}: how should the firm assign workers to tasks so as to complete a given set of tasks with labor, at minimal cost? The value of this program defines the firm's \textit{labor-cost}, as a function of the aggregate allocation of tasks to labor. The main methodological contribution of the paper is a solution to this mechanism-design program (\Cref{thm:concave} and \Cref{thm:outer}). I develop a graphical ironing technique for solving the dual to this program, which is used extensively in the subsequent analysis of automation, and which may be of independent interest. In general, the optimal mechanism can be implemented using a two-step protocol in which workers first select from a menu of assignments (how many tasks of each type they must complete) which guarantees them exactly their reservation workload, and can then exchange one type of task for the other according to a nonlinear exchange-rate schedule. Under a regularity condition, this schedule involves at most two distinct rates.  

Armed with the solution to the firm's labor-market problem, and the corresponding labor-cost function, we can then study the effects of changes in machine technology. Such changes shift the size and composition of the set of tasks allocated to labor, which in turn affects the division of tasks among worker groups, and the mechanisms used to assign tasks to individual workers. The main result on automation, \Cref{thm:no_disposal}, identifies qualitative differences between the effects of an increase automation, in the form of a reduced aggregate allocation of tasks to labor, depending on the degree of asymmetric information. In the low-asymmetry case, which is defined by conditions on the distribution of workers' preferences, an increase in automation reduces the firm's marginal cost of completing both tasks, reduces the surplus generated by every group of workers, and shifts the allocations across groups such that no group receives more of both types of tasks. On the other hand, if the low-asymmetry condition is violated, there are increases in automation which increase the firm's marginal cost of completing one of the tasks, increase the surplus generated by every group, and strictly increase the allocation of both tasks for some groups. The key to establishing \Cref{thm:no_disposal} is an intermediate result, \Cref{prop:sub_dispersion}, which implies that the firm's labor-cost function is supermodular (meaning tasks are substitutes in the cost-minimization problem) if and only if information asymmetries are low. 

I then study the set of aggregate allocations of tasks to labor that could arise when the firm chooses output levels and computational resources optimally to maximize profit. Again, there is a qualitative difference between the low- and high-asymmetry cases. In general, the high-asymmetry case favors a more moderate labor allocation at the aggregate level, in that the firm assigns a more even mix of tasks to labor (\Cref{prop:eqm_properties}). Nonetheless, I show that as technology evolves over time, the allocations of \textit{individual} workers become more  specialized, regardless of the degree of asymmetry (\Cref{prop:specialization}).\footnote{However, some individual workers may still receive some of both types of tasks when there is some asymmetry, whereas assignments under symmetric information are fully specialized.} This comparison highlights the potential differences between individual- and aggregate-level changes. 

Finally, to emphasize that the predictions of the model are driven primarily by information asymmetries rather than the firm's market power, I extend the model to cover multiple firms which compete for workers. This is a mechanism-design problem with competing principals, in which the mechanism offered by each firm serves as an (endogenous) outside option for workers contracting with any other firm. Nonetheless, I show that the mechanisms which were optimal in the single-firm case remain equilibrium outcomes (\Cref{thm:competitive}).

For applied work, the main takeaway is that asymmetric information about workers' preferences matters in task-allocation problems, and can reverse the predictions regarding the effect of automation on the labor market. While the simple model studied here omits a number of important real-world factors (such as strategic wage setting by firms, demand effects, and production using more than two tasks) the results suggest that failure to account for workers' private information could bias the predictions of richer models, and that further work to incorporate information asymmetries into the more general task-allocation frameworks used in empirical applications could be valuable (e.g., \citealt{acemoglu2018race, korinek2024scenarios, acemoglu2024tasks}).

\vst
\noindent\textit{Related literature} 
\vst

Conceptually, this paper builds on the large literature that models production as a task- allocation problem \citep{roy1951some, zeira1998workers, acemoglu2001productivity}. The modern strand of this literature applied to automation was initiated by \cite{acemoglu2018race}, who combine a task-based model of the labor market with a model of directed technological change. Within this literature, the paper that is closest in spirit to the current focus on information asymmetries is \cite{acemoglu2024automation}, who study automation in markets with distortions that cause wages to be above opportunity costs. Such distortions are distinct from the information rents studied here; in \cite{acemoglu2024automation} distortions take the form of fixed wedges between wages and opportunity costs, whereas here rents are determined endogenously as part of the optimal mechanism.\footnote{The reduced-form model of \cite{acemoglu2024automation} admits micro-foundations based on efficiency wage considerations, regulations, or bargaining.} Such wedges amplify the displacement effect of automation on workers, but without a countervailing incentive effect that arises under asymmetric information. A key insight of the current paper is that these incentive effects can play an important role. 

Methodologically, this paper contributes to the mechanism-design literature. In the language of mechanism design, the single-firm problem is one of large-market, multi-item allocation and no transfers. Moreover, the problem features type-dependent participation constraints, which makes it one of so-called countervailing incentives \citep{lewis1989countervailing, maggi1995countervailing, jullien2000participation}.

In terms of the problem studied, this paper is most closely related to \cite{baron2025mechanism}, which studies the assignment of Children's Protective Services investigators to cases. There the authors consider a special case of the general model studied here, in which the participation constraint is determined by a random allocation mechanism used in the status quo. They introduce the two-stage procedure used here to characterize the optimal mechanism, in which an ``inner program'' is used to characterize the incentive constraints of agents, and an ``outer program'' solves the principal's program by studying its dual. However, their technique for solving the inner program does not produce an explicit solution, but rather identifies qualitative features of the optimal mechanism. Here instead we give a more explicit solution to the inner program by developing a general ironing technique based on \cite{myerson1981optimal}. This technique plays a key role in establishing the properties and comparative statics of automation.

The ironing technique used here to solve the inner program is similar to that introduced by \cite{dworczak2024mechanism} to study optimal property rights. While the problem that they study differs from the inner program here in a number of dimensions, an important technical difference is that here the principal allocates non-negative quantities of two tasks, whereas there the principal allocates a physical good which takes values in $[0,1]$, and transfers which can be any real number. As a result, while both \cite{dworczak2024mechanism} and the current paper build on \cite{myerson1981optimal}, the solutions do not coincide.\footnote{Additionally, \cite{dworczak2024mechanism} focus on the case of a convex outside option (which is analogous to the concave case in the current setting) whereas here the outside option is unrestricted.} 

The model with multiple-firms also contributes to the literature on competing principals (e.g., \citealt{martimort1996exclusive, ellison2004competing}). Here the participation constraint in the mechanism-design problem is determined by the equilibrium contracts offered by firms.  

\Cref{sec:model} presents the baseline model with a single firm. \Cref{sec:solution} solves the firm's mechanism-design problem. \Cref{sec:automation} presents the implications for automation. \Cref{sec:competitive} extends the model to include multiple firms competing for workers. \Cref{sec:conclusion} concludes. Auxiliary mechanism-design results are deferred to \Cref{sec:general}. Proofs omitted from the main body are collected in \Cref{sec:omitted_proofs}.

\section{Model}\label{sec:model}

The basic idea of the production-as-task-allocation framework is that aggregate output is produced by combining the services of a set of tasks. Following \cite{acemoglu2018race}, much of the literature has considered constant elasticity of substitution (CES) aggregation of a continuum of tasks, which enables a tractable analysis of general equilibrium effects. I instead consider only two types of tasks. This is the simplest model in which workers' preferences over tasks can be represented, and allows me to focus on the role of information asymmetries.\footnote{Indeed, allowing for more than two tasks introduces well-known technical challenges related to multi-dimensional screening (e.g., \cite{rochet1987necessary, hart2015maximal, lahr2024extreme}), which are beyond the scope of the current study to resolve, and are not directly related to the main conceptual question of interest.} 

In the baseline model I focus on a market with a single monopsony firm, and take the set of workers employed by the firm, and their wages, as fixed. This precludes the study of many interesting questions related to wage dynamics, which are central to the debate on the effects of automation. The model is thus better suited to understanding short-run dynamics within the firm, or markets where frictions, such as collective bargaining agreements or minimum-wage regulations, prevent wages from adjusting. Understanding the role of asymmetric information in the fixed-wage case is also a necessary first step towards incorporating these frictions into a richer model with wage adjustments. I take a small step in this direction in \Cref{sec:competitive}, where I extend the model to accommodate multiple firms competing for workers, and allow wages to adjust in equilibrium. The single-firm model is easier to describe, which motivates the choice to present it first. 

\vst
\noindent\textbf{Setup:} Formally, the baseline (monopsony) model is the following. There are two types of tasks, $A$ and $B$. The firm's aggregate output of the consumption good is given by $Y(\mu^A,\mu^B)$, where $\mu^k$ is the total units of $k$ tasks completed, and $Y$ is concave and strictly increasing. The inverse demand faced by the firm for the final consumption good is $P$.

\vst
\noindent\textbf{Production:} As in \cite{acemoglu2018race}, tasks can be completed by the firm using either a worker or a machine. The firm employs a continuum of workers, each belonging to one of a discrete set of groups indexed by $j \in \mathcal{J} := \{1, \dots, J\}$. Without loss of generality, assume that each group consists of a unit mass of individuals. As in \cite{acemoglu2024automation}, workers within the same group share the same productivity across tasks, which is observed by the firm. Productivity is represented by $\pi_j^k > 0$, the marginal cost to the firm of assigning a group $j$ worker to complete one unit of task $k$. This cost represents, for example, the time and company resources used by the worker to complete the task.\footnote{Equivalently, we could think of $\pi_j^k$ as $1/\rho_j^k$, where $\rho_j^k$ is the worker's marginal productivity per unit of time spent on task $k$. I formulate performance in terms of cost-per-completed-unit for convenience.}

In addition, each worker is characterized by their preferences over tasks. We represent a worker's preferences by their \textit{type} $(\theta^A,\theta^B)$, where $\theta^k > 0$ is the worker's marginal cost for completing task $k$. The \textit{subjective workload} of a worker assigned to complete $n^A$ units of task $A$ and $n^B$ units of task $B$ is thus given by $\theta^A n^A + \theta^B n^B$. A worker's type is their private information. The distribution of types is allowed to vary across worker groups; denote the CDF of the type distribution in group $j$ by $F_j$, and its support by $\Theta_j$.

Tasks can also be completed by a machine. Let $\pi_m^k > 0$ be the performance of the machine on task $k$, modeled in the same way as that of workers. In addition, employing the machine to complete a unit of task $k$ requires $c^k > 0$ units of a resource I refer to as ``computation''. If the firm allocates to the machine $n^A$ units of task $A$ and $n^B$ of task $B$ then it uses $c^A n^A + c^B n^B$ units of computation. Let $\gamma: \R_+ \rightarrow \R$ be the cost of computation, assumed to be strictly increasing and convex. 

\vst
\noindent\textbf{The firm's program: } 
The firm's profit-maximization program can be stated as
\[
    \max_{\mu^A,\mu^B} P\left(Y(\mu^A,\mu^B) \right) Y(\mu^A,\mu^B) - C(\mu^A,\mu^B)
\]
where $C(\mu^A,\mu^B) := \min_{a_m,b_m} \left\{\gamma\left(c^A a_m + c^B b_m \right) + \pi^A_m a_m + \pi^B_m b_m  + L(\mu^A - a_m,\mu^B - b_m)\right\},$ and $L(\ell^A, \ell^B)$ is the cost of completing tasks $(\ell^A,\ell^B)$ using labor.

The main challenge is to characterize the \textit{labor-cost} function $L$. To do this, we solve for the firm's optimal labor-market mechanism. A mechanism consists of a menu of jobs offered to each group, where a job is a pair $(\tilde{n}^A,\tilde{n}^B) \in \R_+^2$ describing the units to be performed of each task. Equivalently, the firm chooses for each group $j \in \mathcal{J}$ a direct mechanism $(a_j,b_j): \Theta_j \rightarrow \R^2_+$ mapping preferences to bundles of tasks \citep{myerson1981optimal}.\footnote{Note that each agent's allocation depends only on their own type, rather than the entire profile of types. This is feasible because there are a continuum of agents, and thus no aggregate uncertainty about the type profile. In other words, we can work directly with the interim allocation rule, and the constraints to ensure that this is feasible, \`a la \cite{border1991implementation}, are automatically satisfied.}  The firm's labor-market program is 
\begin{align}
    & L(\ell^A, \ell^B) := \min_{(a_j,b_j)_{j=1}^{J}}  \sum_{j=1}^J \E_j\left[ a_j(\theta^A,\theta^B) \pi_j^A + b_j(\theta^A,\theta^B) \pi_j^B \right] \quad \quad s.t. \label{prog:1} \\
     &\theta^A a_j(\theta^A,\theta^B) + \theta^B b_j(\theta^A,\theta^B) \leq \theta^A a_j(\hat{\theta}^A,\hat{\theta}^B) + \theta^B b_j(\hat{\theta}^A,\hat{\theta}^B) \ \ \forall \ (\theta^A,\theta^B), (\hat{\theta}^A,\hat{\theta}^B) \in \Theta_j \tag{IC} \\
    & \theta^A a_j(\theta^A,\theta^B) + \theta^B b_j(\theta^A,\theta^B) \leq r_j(\theta^A,\theta^B) \quad \forall \ (\theta^A,\theta^B) \in \Theta_j \tag{IR}\\
    &\sum_{j \in \mathcal{J}} \E_j\big[ a_j(\theta_j) \big] = \ell^A \quad \text{and} \quad \sum_{j \in \mathcal{J}} \E_j\big[ b_j(\theta_j) \big] = \ell^B \tag{MC}
\end{align}
The (IC) constraint ensures that workers receive a lower workload by reporting truthfully than misreporting. The right hand side of the (IR) constraint is the worker's reservation value, i.e., the maximum workload that can be assigned to a worker, for example due to pre-determined contractual terms or concerns about burnout. We allow for the possibility that the reservation value depends on the worker's type, and assume only that $r_j$ is bounded and non-negative.\footnote{For example, the reservation value could be the workload from some default allocation of tasks that the worker has the right to demand. In addition to the descriptive appeal of allowing for a type-dependent outside option, the generality of allowing for type-dependent reservation values is useful when it comes to modeling the competitive market, as in \Cref{sec:competitive}, where $r_j$ is determined endogenously via the contracts offered by other firms.} The market-clearing (MC) constraint ensures that the target levels of task completion are achieved. As this program makes clear, the fundamental difference between machines and workers in the model is that the latter possesses private information about their preferences: fixing a ``computing budget'' $\bar{c}$, a machine is just a worker with a degenerate type distribution and a reservation value $\bar{c}$.

\section{Solving the firm's program}\label{sec:solution}

To begin, observe that workers' preferences over assignments are fully determined by $\theta := \theta^A/\theta^B$, the relative cost of task $A$. Moreover, define the effective reservation value 
\[
R_j(\theta) := \min\left\{ \frac{1}{\theta^B} r_j(\theta^A,\theta^B) : (\theta^A,\theta^B) \in \Theta_j \ , \  \frac{\theta^A}{\theta^B} = \theta \right\}.
\]

\begin{lemma}\label{lem:effective_R}
    If mechanism $(a_j,b_j)$ is IC then it satisfies IR for $j$ if and only if $\frac{\theta^A}{\theta^B} a_j(\theta^A,\theta^B) + b_j(\theta^A,\theta^B) \leq R_j(\theta^A/\theta^B)$ for all $(\theta^A,\theta^B) \in \Theta_j$.
\end{lemma}
\begin{proof}
    Proof in \Cref{proof:effective_R}.
\end{proof}

As a result, it is without loss of optimality to restrict attention to mechanisms which condition only on $\theta = \theta^A/\theta^B$, with reservation value $R_j(\theta)$.\footnote{\cite{dworczak2021redistribution} make a similar observation in a model of monopoly-pricing. See also \cite{baron2025mechanism}.} As shown below, it is also without loss of generality to take $R_j$ to be non-decreasing. Where it will not cause confusion, I simply refer to the ratio $\theta = \theta^A/\theta^B$ as the agent's type. Let $F_j$ be the cdf of $\theta$ in group $j$. We maintain the following assumption.

\begin{assumption}\label{ass:distribution}
    For all $j$, the distribution $F_j$ has full support on and interval $[\und{\theta}_j,\bar{\theta}_j]$ with $\und{\theta}_j > 0$, and admits a strictly positive density $f_j$.
\end{assumption}

In light of \Cref{lem:effective_R}, we restate the firm's labor-market program as choosing a mechanism consisting of functions $(a_j,b_j): [\und{\theta}_j,\bar{\theta}_j] \rightarrow \R^2_+$ for each $j \in \mathcal{J}$ to solve
\begin{align}
\min_{(a_j,b_j)_{j=1}^J} \sum_{j \in \mathcal{J}} \mathbbm{E}_{j}\left[ a_j(\theta_j)\pi^A_j + b_j(\theta_j) \pi^B_j \right]& \label{prog:2}\\
  s.t. \quad  a_j(\theta_j)\theta_j + b_j(\theta_j) \leq a_j(\theta'_j)\theta_j + b_j(\theta'_j)  \quad\quad &\forall \ j \in \mathcal{J}, \ \ \theta_j, \theta'_j \in \Theta_j \tag{IC} \\
    a_j(\theta_j)\theta_j + b_j(\theta_j)  \leq R_j(\theta_j) \quad\quad &\forall \ j \in \mathcal{J}, \ \theta_j \in \Theta_j \tag{IR} \\
    \sum_{j \in \mathcal{J}} \E_j\big[ a_j(\theta_j) \big] = \ell^A \quad \text{and}& \quad \sum_{j \in \mathcal{J}} \E_j\big[ b_j(\theta_j) \big] = \ell^B
    \label{eq:MC}. \tag{MC}
\end{align}

\subsection{Two-step approach}

We solve the program in \cref{prog:2} in two steps, following the approach of \cite{baron2025mechanism}. Observe that both the firm's objective and the market-clearing conditions depend only on the aggregate allocations for each group. Let $\mathcal{F}_j \subset \R^2_+$ be the set of pairs $\left(\E_j\left[ a_j(\theta_j)\right], \E_j\left[ b_j(\theta_j)\right]\right)$, i.e., aggregate allocations for group $j$, that can be induced by some IR and IC mechanism. We call such pairs \textit{incentive feasible}. Observe that $\mathcal{F}_j$ is compact and convex, as it is defined by a set of linear inequalities (the IC and IR constraints).

Suppose for a moment that we knew the incentive-feasible set for each group. Then we could solve an ``outer program''
    \begin{align*}
      L(\ell^A,\ell^B) \ := \ \min_{(n^A_j, n^B_j)_{j\in \mathcal{J}}} \sum_{j \in \mathcal{J}} n_j^A \pi_j^A + n_j^B \pi_j^B \quad\quad s.t. \quad &(n_j^A, n_j^B) \in \mathcal{F}_j \ \ \forall \ j \in \mathcal{J} \\
        &  \sum_{j \in \mathcal{J}} n_j^A = \ell^A \\
        & \sum_{j \in \mathcal{J}} n_j^B = \ell^B
    \end{align*}
We can see that the program only depends on the agents' outside options and the distribution of preferences via the incentive-feasible sets $\mathcal{F}_j$.\footnote{Having characterized these sets for each agent, the outer program is similar to that studied by \cite{baron2025mechanism} (with some small differences regarding the nature of the market-clearing constraint). The primary difference between the current study and \cite{baron2025mechanism} lies in the structure of the inner program (and the approach to solving the inner program). As a result, the optimal mechanisms can look quite different, even though the outer programs appear similar.} The constraint that $(n^A_j,n^B_j)$ be incentive feasible ensures that this outcome can be implemented by some IC and IR mechanism, and is thus a solution to the principal's cost-minimization program.

Let $\lambda^k$ be the multiplier on the market-clearing constraint for task $k$. Then the outer program is equivalent to 
\begin{equation}\label{prog:outer_primal}
    \min_{(n^A_j, n^B_j) \in \mathcal{F}_j} \ \sup_{\lambda^A,\lambda^B} \ \ \sum_{j \in \mathcal{J}} (n_j^A \pi_j^A + n_j^B \pi_j^B) + \lambda^A\left(\ell^A -  \sum_{j \in \mathcal{J}} n_j^A \right) + \lambda^B\left(\ell^B -  \sum_{j \in \mathcal{J}} n_j^B \right)
\end{equation}
The dual to this program can be written as 
\begin{align}
     \sup_{\lambda^A,\lambda^B} &\min_{(n^A_j, n^B_j) \in \mathcal{F}_j}  \ \lambda^A \ell^A + \lambda^B \ell^B - \sum_{j \in \mathcal{J}} \left( \left( \lambda^A - \pi_j^A \right)n_j^A + \left( \lambda^B - \pi_j^B \right)n_j^B \right) \label{prog:outer_dual1} \\
     &= \sup_{\lambda^A,\lambda^B} \ \ \lambda^A \ell^A + \lambda^B \ell^B - \sum_{j\in \mathcal{J}} S_j\left(  \lambda^A - \pi_j^A  \ , \ \lambda^B - \pi_j^B \right) \label{prog:outer_dual2}
\end{align}
where 
\begin{equation}\label{eq:support}
    S_j(w^A,w^B) := \max_{(n^A,n^B) \in \mathcal{F}_j } w^A n^A + w^B n^B.
\end{equation}
is the support function of the convex set $\mathcal{F}_j$.\footnote{The support function $S_j$ yields the dual representation 
\[
\mathcal{F}_j = \{(n^A,n^B) : w^A n^A + w^B n^B \leq S_j(w^A,w^B) \ \forall \ (w^A,w^B) \in \R^2 \}.
\]} Let $(N^A(w^A,w^B), N^B(w^A,w^B)) \subset \R_+^2$ be the solutions to the program \cref{eq:support}. We show that strong duality holds, i.e., the value of the primal equals that of the dual, so we can obtain the optimal mechanism by solving the program in \cref{prog:outer_dual2}, where $(\lambda^A,\lambda^B)$ are shadow prices on the market-clearing constraints, i.e., the marginal costs of completing tasks $A$ and $B$ with labor.

\begin{proposition}\label{thm:outer}
     If the principal's program in (\ref{prog:2}) is feasible then strong duality for the outer program holds. Moreover, if $\lambda_*^A,\lambda_*^B$ is a solution to the dual outer program in (\ref{prog:outer_dual2}) then there exist selections $(n^A_j,n^B_j) \in N^*_j(\lambda_*^A - \pi_j^A,\lambda_*^B - \pi_j^B)$ which satisfy market clearing, and these constitute a solution to the primal outer program. 
\end{proposition}
\begin{proof}
    Proof in \Cref{proof:outer}.
\end{proof}

Define the \textit{frontier} of $\mathcal{F}_j$ to be the set 
\[
\{ N_j^*(w^A,w^B) \in \mathcal{F}_j \ : \  (w^A,w^B) \not\leq 0 \}
\] 
Thus $\mathcal{F}_j$ is the subset of the positive orthant enclosed by the frontier. Abusing terminology, we say that the \textit{frontier of $\mathcal{F}_j$ is strictly convex} if any convex combination of points on the frontier lies in the interior of $\mathcal{F}_j$. When $F_j = F_{j'}$ and $R_j = R_{j'}$, so that $\mathcal{F}_j = \mathcal{F}_{j'}$ for all $j,j'$ the solution to the outer problem is easy to visualize. \Cref{fig:outer} depicts the frontier of the incentive-feasible set. The dots along the frontier are the allocations for each of five groups. The market clearing condition requires precisely that these points have barycenter $(\ell^A/J, \ell^B/J)$. The arrows depict the direction and magnitude of $(\lambda_*^A - \pi_j^A, \lambda_*^B - \pi_j^B)$ for each $j$. Notice that the frontier, as depicted in \Cref{fig:outer}, can have upward-sloping segments. This possibility plays an important role in shaping the dynamics of automation, and is discussed in detail in the subsequent section. 

From \Cref{thm:outer} we can read off some properties of the optimal mechanism. Say that group $j$ is on their frontier if their aggregate allocation is in this set, and refer to a group $j$ which is not on their frontier as \textit{remedial} (all groups are on the frontier in \Cref{fig:outer}).

\begin{corollary}\label{cor:outer}
    All optimal mechanisms have the following properties
    \begin{enumerate}[i.]
        \item Each agent in group $j$ is offered a mechanism which solves the ``inner program'' for weights $(w^A,w^B) = (\lambda^A_* - \pi_j^A, \lambda^B_* - \pi_j^B)$ (see \Cref{sec:inner}). 
        \item The allocation of any agent in a remedial group, $j$, does not depend on their type if $(\lambda^A_* - \pi_j^A, \lambda^B_* - \pi_j^B) \neq (0,0)$.
        \item If no two groups have identical performance on any task (i.e., $\pi_j^k \neq \pi_{j'}^k$ for all $j,j' \in \mathcal{J}$, $k\in \{A,B\}$) then there are at most two remedial agents with non-zero allocations. If there are two such agents then one receives only task $A$ and the other only task $B$.
    \end{enumerate}
    Moreover, there exists a solution in which for any group $j$ there are weights $(w^A_j , w^B_j)$ with at least one dimension strictly positive, and a scalar $\beta_j \in [0,1]$, such that the mechanism offered to agents in this group is $\beta_j \cdot (a^*(\theta), b^*(\theta)) $, where $(a^*(\theta), b^*(\theta))$ solves the inner program for weights $(w^A_j , w^B_j)$. If $j$ is on their frontier then $\beta_j = 1$ and $(w^A_j , w^B_j) = (\lambda^A_* - \pi_j^A, \lambda^B_* - \pi_j^B)$.
\end{corollary}
\begin{proof}
    Part \textit{i.} is immediate from \Cref{thm:outer} and the definition of the inner program in \Cref{sec:inner}. Part \textit{ii.} follows because any such agent either receives at most one type of task, or is excluded entirely. Part 3 holds because under non-identical performance there is at most one agent with $\lambda^A_* - \pi_j^A = 0$, and at most one with $\lambda^B_* - \pi_j^B = 0$. The final point follows because any point in the set $\mathcal{F}_j$ can be induced in this way. 
\end{proof}

In summary, if we know support functions $S_j$ of each group's incentive-feasible set $\mathcal{F}_j$, the preceding argument characterizes the labor cost function as
\[
L(\ell^A,\ell^B) = \sup_{\lambda^A,\lambda^B} \ \left\{ \lambda^A \ell^A + \lambda^B \ell^B - \sum_{j\in \mathcal{J}} S_j\left(  \lambda^A - \pi_j^A  \ , \ \lambda^B - \pi_j^B \right)\right\}
\]
and tells us the optimal aggregate allocations, $\left(\E_j\left[ a_j(\theta_j)\right], \E_j\left[ b_j(\theta_j)\right]\right)$, for each group. To complete the characterization, and to identify the optimal labor market mechanisms (as opposed to just the group-level allocations) we must solve an ``inner program'' to characterize the support functions. However, the main results regarding automation can be stated without direct reference to the inner program, and so readers primarily interested in these results can skip directly to \Cref{sec:automation} before returning to \Cref{sec:inner} for the proofs. 

\begin{figure}
    \centering

    \begin{tikzpicture}[scale = .4, >=latex, font=\sffamily]

\draw[->, line width=0.7pt] (0,0) -- (20.3,0) node[below, yshift=-3pt]{$A$};
\draw[->, line width=0.7pt] (0,0) -- (0,16.7) node[left, xshift=-4pt]{$B$};

\draw[densely dotted, thick] (0,0) -- (14,0);
\draw[densely dotted, thick]
  (14,0) .. controls (16.7,8.8) and (10.7,14.9) .. (5.2,15.6)
           .. controls (3.2,15.6)  and (1.6,14.4)  .. (0,11.1);
\draw[densely dotted, thick] (0,11.1) -- (0,0);

\node[circle, draw=none, fill=none] (mu) at (7.5,9.7) {};
\node at (9,9.9) {$\circ$};
\node[anchor=north] at (7.5,9.7) {\small $\left(\frac{\ell^A}{J},\frac{\ell^B}{J}\right)$};

\coordinate (P1) at (1.5,13.6);
\coordinate (P2) at (3.2,15.2);
\coordinate (P3) at (8,14.8);
\coordinate (P4) at (14.7,5.6);
\coordinate (P5) at (14.0,0.0);

\foreach \p/\n in {P1/1,P2/2,P3/3,P4/4,P5/5}{
  \fill (\p) circle (5pt);
}

\draw[->, line width=0.9pt] ($(P1)$) -- ($(P1)+(-0.75,0.7)$);
\draw[->, line width=0.9pt] ($(P2)$) -- ($(P2)+(-0.9,1.6)$);
\draw[->, line width=0.9pt] ($(P3)$) -- ($(P3)+(0.7,1.6)$);
\draw[->, line width=0.9pt] ($(P4)$) -- ($(P4)+(3.2,0.32)$);
\draw[->, line width=0.9pt] ($(P5)$) -- ($(P5)+(0.9,-0.9)$);

\end{tikzpicture}

    \caption{The outer program}
    \label{fig:outer}
\end{figure}
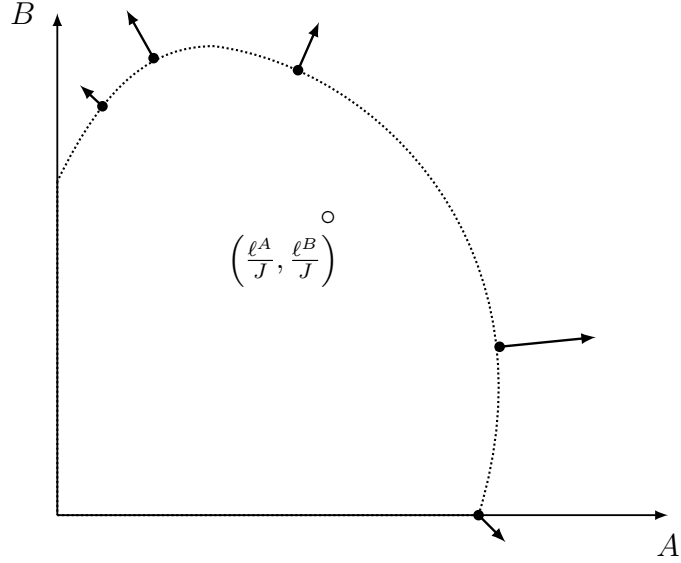

\subsection{The inner program}\label{sec:inner}

The goal of the inner program is to characterize the set $\mathcal{F}_j$ of incentive-feasible pairs for a given group via the support function $S_j$. This program considers each group separately, so for convenience we suppress the $j$ subscript in the notation. As the set $\mathcal{F}$ is defined by the IC, IR, and non-negativity constraints on the mechanism $(a,b)$, we characterize $S(w^A,w^B)$ by solving the maximization problem in \cref{eq:support} directly in terms of the underlying mechanism, subject to these constraints. That is, we solve
\begin{align}
 S(w^A,w^B) \ =& \max_{a,b} \ \int w^A a(\theta) + w^B b(\theta) dF(\theta) \label{eq:inner} \\
         s.t. \quad & \theta a(\theta) + b(\theta) \leq \theta a(\theta') + b(\theta') \quad \quad \forall \ \theta,\theta' \in [\und{\theta},\bar{\theta}] \tag{IC} \\
        &\theta a(\theta) + b(\theta) \leq R(\theta) \quad \quad \forall \ \theta \in [\und{\theta},\bar{\theta}] \tag{IR}\\
        & a,b \geq 0.\notag
\end{align}

This program now bears some similarity to a monopoly pricing problem \citep{mussa1978monopoly, myerson1981optimal}, with the allocation of task $B$ playing the role of transfers. Most closely related to the current program are \cite{dworczak2024mechanism} and \cite{baron2025mechanism}. Translated into the current notation, \cite{dworczak2024mechanism} solve a modification of \cref{eq:inner} in which $R$ is concave, $a$ is bounded above, and $b$ is unconstrained. In \cref{eq:inner} on the other hand, the only upper bounds on $a$ and $b$ come endogenously through the IR constraint. This necessitates a different solution technique, even for the special case of \cref{eq:inner} in which $R$ is concave. \cite{baron2025mechanism} study a special case of \cref{eq:inner} in which $R$ is linear, and employ an alternative solution technique. Here we develop an alternative graphical approach to solve the problem for general $R$. This approach is both simpler and facilitates comparative statics. 

To solve \cref{eq:inner}, we first simplify the program using the envelope theorem \citep{milgrom2002envelope} and the Spence-Mirrlees characterization of incentive compatibility: by the envelope theorem
\[
a(\theta)\theta +b(\theta) = a(\und{\theta})\und{\theta} + b(\und{\theta}) + \int_{\und{\theta}}^\theta a(z)dz
\]
so 
\[
b(\theta) = -a(\theta)\theta + \und{u} + \int_{\und{\theta}}^\theta a(z)dz.
\]
for some $\und{u} \geq 0$. Integrating by parts we have
\[
    \int_{\und{\theta}}^{\bar{\theta}} \int_{\und{\theta}}^\theta a(z)dz \ dF(\theta) = \int_{\und{\theta}}^{\bar{\theta}}a(\theta)d\theta - \int_{\und{\theta}}^{\bar{\theta}} a(\theta)F(\theta) d\theta. 
\]
Then substituting into the objective, the program becomes
\begin{align}
    \max_{a:[\und{\theta},\bar{\theta}] \rightarrow \R_+, \ \und{u} \geq 0} &\int_{\und{\theta}}^{\bar{\theta}} a(\theta)\left(w^A - w^B \left(\theta - \frac{1 - F(\theta)}{f(\theta)} \right) \right) f(\theta) d\theta  + w^B \und{u}  \label{prog:alpha}\\
    s.t. \quad & a \text{ non-increasing } \tag{IC}\\
    & \und{u} + \int_{\und{\theta}}^\theta a(z)dz \leq R(\theta) \quad \forall \ \theta \tag{IR} \\
    & -a(\theta)\theta + \und{u} + \int_{\und{\theta}}^\theta a(z)dz \geq 0 \quad \forall \ \theta. \tag{NN}
\end{align}

Because $a \geq 0$ we have the following observation, referenced above, that it is without loss of generality to assume that $R$ is non-decreasing.  

\begin{lemma}\label{lem:R_nondecreasing}
The program is equivalent if we replace $R$ with its non-decreasing lower envelope (the largest non-decreasing function bounded above by $R$).     
\end{lemma}
  
Because $a$ is non-increasing it suffices to impose the non-negativity constraint on type $\und{\theta}$. Then we can simplify the program to 
\begin{align}
    \max_{a:[\und{\theta},\bar{\theta}] \rightarrow \R_+, \ \und{u} \geq 0} &\int_{\und{\theta}}^{\bar{\theta}} a(\theta) W(\theta) d\theta  + w^B \und{u} \label{prog:relax0} \\
    s.t. \quad & a \text{ non-increasing } \tag{IC}\\
    & \und{u} + \int_{\und{\theta}}^\theta a(z)dz \leq R(\theta) \quad \forall \ \theta \tag{IR} \\
    & -a(\und{\theta})\und{\theta} + \und{u} \geq 0 \label{eq:nn} \tag{NN}
\end{align}
where $W(\theta) := \left(w^A - w^B \left(\theta - \frac{1 - F(\theta)}{f(\theta)} \right)\right) f(\theta)$. Define 
\begin{equation}\label{eq:mathcalW}
\begin{split}
\mathcal{W}(\theta) &:= \int_{\und{\theta}}^\theta W(z) dz\\
&= w^A F(\theta) - w^B \left( \int_{\und{\theta}}^{\theta} z f(z)dz - (\theta - \und{\theta}) + \int_{\und{\theta}}^{\theta} F(z)dz \right) \\ 
&= w^A F(\theta) - w^B \left(\und{\theta}  - \theta(1 - F(\theta)) \right)
\end{split}
\end{equation}
where the final equality follows by integrating by parts. Let $\bar{\mathcal{W}}$ be the concavification (concave upper envelope) of $\mathcal{W}$ and let $\bar{W}(\theta) = \bar{\mathcal{W}}'(\theta)$. We say that the program is \textit{regular} if $\bar{\mathcal{W}} = \mathcal{W}$.

The strategy is to replace $W$ with $\bar{W}$ and look for a solution that satisfies the pooling property, as in \cite{myerson1981optimal}.  That is, we first consider the program 
\begin{align}\label{prog:relax}
    \max_{a:[\und{\theta},\bar{\theta}] \rightarrow \R_+, \ \und{u} \geq 0} &\int_{\und{\theta}}^{\bar{\theta}} a(\theta) \bar{W}(\theta) d\theta  + w^B \und{u} \\
    s.t. \quad & a \text{ non-increasing } \tag{IC}\\
    & \und{u} + \int_{\und{\theta}}^\theta a(z)dz \leq R(\theta) \quad \forall \ \theta \tag{IR} \\
    & -a(\und{\theta})\und{\theta} + \und{u} \geq 0 \tag{NN}
\end{align}
We distinguish between three cases, determined by the shape of $R$. We focus here on the case where $R$ is concave. This holds for example if $R$ is the indirect utility attained by choosing task allocations from some menu, as is the case in many applications. More generally, notice that the left hand side of the IR constraint must be a concave function, since $a$ is non-increasing. Thus whether or not $R$ is concave, there exists some concave $\hat{R}$ with $\hat{R}(\theta) \leq R(\theta)$ for all $\theta$, such that the solution is unchanged if we replace $R$ with $\hat{R}$. Thus we can reduce the program to the concave case. If $R$ is convex then we can solve for $\hat{R}$ in closed form. Otherwise, in general we can only describe some properties that it must have. For brevity, the details of the convex and general cases are deferred to \Cref{sec:general}.

\subsubsection{The concave case}

Assume that $R$ is concave and non-decreasing. Define
\begin{align*}
\und{\theta}^* &:= \min\{ \argmax_{\theta \in [\und{\theta},\bar{\theta}]} \ \mathcal{W}(\theta) \} \\
\bar{\theta}^* &:= \max\{ \argmax_{\theta \in [\und{\theta},\bar{\theta}]} \ \mathcal{W}(\theta) \}
\end{align*}
Observe that $\und{\theta}^* > \und{\theta}$ if and only if $\max \mathcal{W}(\theta) > 0$ (because $\mathcal{W}(\und{\theta}) = 0$). Moreover, $\bar{W} > (<) \ 0$ if and only if $\theta < \und{\theta}^*$ ($\theta > \bar{\theta}^*$). Also $\und{\theta}^* = \min\{ \argmax_{\theta \in [\und{\theta},\bar{\theta}]} \ \bar{\mathcal{W}}(\theta) \} $ and $\bar{\theta}^* = \max\{ \argmax_{\theta \in [\und{\theta},\bar{\theta}]} \ \bar{\mathcal{W}}(\theta) \}$.

\begin{lemma}\label{lem:fixed_u}
    For a fixed $\und{u} \in [0,R(\und{\theta})]$, all and only solutions (up to zero-measure perturbations) to the program in (\ref{prog:relax}) take the following form: 
    \begin{equation*}
        a^*(\theta) = 
        \begin{cases}
        \frac{\und{u}}{\und{\theta}} \quad &\text{on } [\und{\theta},  \min\{\tilde{\theta}, \und{\theta}^*\}] \\
        R'(\theta) &\text{on } (\tilde{\theta}, \und{\theta}^*] \\
        \text{non-increasing and bounded above by } R'(\und{\theta}^*) & \text{on } (\und{\theta}^*,\bar{\theta}^*] \\
        0 &\text{on } (\bar{\theta}^*, \bar{\theta}]
        \end{cases}
    \end{equation*}
    where $\tilde{\theta}$ is the unique solution to
    \begin{equation}\label{eq:tilde_c}
        \tilde{\theta} \frac{\und{u}}{\und{\theta}} = R(\tilde{\theta}).
    \end{equation}
    if this is below $\bar{\theta}$, or $\bar{\theta}$ otherwise.
\end{lemma}
A simple proof of \Cref{lem:fixed_u} follows from well-known results in combinatorial optimization. An elementary alternative proof is given in \Cref{proof:altfixed_u}.\footnote{The result could also be proven by observing that \cref{eq:relax} is essentially a weak majorization constraint, and applying well-known results about maximization subject to such constraints (see \cite{ryff1967extreme}). A version of this approach is taken by \cite{dworczak2024mechanism}, who prove a similar result by drawing a connection to second-order stochastic dominance. See also \cite{kleiner2021extreme}.}
\begin{proof}
    For any $\und{u}$, the non-negativity constraint implies that $-a(\theta)\und{\theta} + \und{u} \geq 0$ for all $\theta$, which in turn implies $\int_{\und{\theta}}^\theta a(z)dz \leq (\theta - \und{\theta})\frac{\und{u}}{\und{\theta}}$ for all $\theta$. Then consider the relaxed program 
    \begin{align}
    \max_{a:[\und{\theta},\bar{\theta}] \rightarrow \R_+} &\int_{\und{\theta}}^{\bar{\theta}} a(\theta) \bar{W}(\theta) d\theta \notag  \\
    s.t. \quad & \und{u} + \int_{\und{\theta}}^\theta a(z)dz \leq \min \left\{R(\theta), \ \und{u} + (\theta - \und{\theta})\frac{\und{u}}{\und{\theta}} \right\} \quad \forall \ \theta. \label{eq:relax}
\end{align}
With a finite set of types, the feasible set defined by \cref{eq:relax} is a polymatroid (see \cite{schrijver2003combinatorial} Section 44.1a, or \cite{che2013generalized} Lemma 1). \cite{Edmonds1970} shows that maximization of a linear objective over a polymatriod is solved by a greedy procedure, yielding the finite-state-space analog to the solution $a^*$.\footnote{See also \cite{schrijver2003combinatorial} Theorem 40.1.} This characterization is easily extended to the infinite-dimensional setting by taking the limit of a sequence of finite grids in $[\und{\theta},\bar{\theta}]$.  
\end{proof}

We know that the solution to the program (\ref{prog:relax}) will take the form in \Cref{lem:fixed_u} for some $\und{u}$. We now want to optimize over $\und{u}$. Alternatively, we can maximize over the threshold $\tilde{\theta}$. To this end, it is convenient to extend $R$ to $[\und{\theta}, \infty)$ by letting $R(\theta) = R(\bar{\theta})$ for $\theta > \bar{\theta}$. Then define 
\[\check{\theta} = \max\left\{\theta \in [\und{\theta},\infty) : R(\theta)/\theta \geq R(\und{\theta})/\und{\theta} \right\}.
\]
Then we can choose a threshold $\tilde{\theta}$ defining $a^*$ as in \Cref{lem:fixed_u}, and define $\und{u} = \und{\theta} R(\tilde{\theta})/\tilde{\theta}$. Program (\ref{prog:relax}) becomes
\begin{align}\label{eq:inner_concave_alt}
    \max_{\theta \geq \check{\theta}} \ \ \frac{R(\theta)}{\theta} \bar{\mathcal{W}}(\min\{\theta, \und{\theta}^*\}) + \int\limits_{\min\{\theta, \und{\theta}^*\}}^{\und{\theta}^*}& \bar{W}(z)R'(z)dz + w^B \und{\theta} \frac{R(\theta)}{\theta} 
\end{align}
Define\footnote{Because $\bar{\mathcal{W}}$ is linear where $\bar{\mathcal{W}} \neq \mathcal{W}$,  $\bar{\theta}^{\#} = \max\left\{ \{\theta \in [\und{\theta},\bar{\theta}] : {\theta}\bar{W}({\theta}) - w^B \und{\theta} - \bar{\mathcal{W}}({\theta}) \geq 0 \} \cup \{\und{\theta}\}\right\}$ and $\und{\theta}^{\#} = \min\left\{ \{\theta \in [\und{\theta},\bar{\theta}] : {\theta}\bar{W}({\theta}) - w^B \und{\theta} - \bar{\mathcal{W}}({\theta}) \leq 0 \} \cup \{\bar{\theta}\}\right\}.$}
\begin{align*}
\bar{\theta}^{\#} &:= \max\left\{ \{\theta \in [\und{\theta},\bar{\theta}] : {\theta}W({\theta}) - w^B \und{\theta} - \mathcal{W}({\theta}) \geq 0 \} \cup \{\und{\theta}\}\right\} \\
\und{\theta}^{\#} &:= \min\left\{ \{\theta \in [\und{\theta},\bar{\theta}] : {\theta}W({\theta}) - w^B \und{\theta} - \mathcal{W}({\theta}) \leq 0 \} \cup \{\bar{\theta}\}\right\}.
\end{align*}
These parameters are illustrated in \Cref{fig:tilde_c}. Combining our previous observations, we have the following characterization of the solution to the program in \cref{prog:relax} and \cref{eq:inner_concave_alt}.

\begin{proposition}\label{prop:concave_relax}
    The solutions to the program in \cref{prog:relax} take the following form:
    \begin{itemize}
        \item If $w^A,w^B \leq 0$ then the value of the program is $0$. A solution is given by $\und{u} = 0$ and $a(\theta) = 0$ for all $\theta$ (and thus $ b(\theta) =0$ for all $\theta$). This solution is unique if $w^A,w^B < 0$. If either $w^A = 0$ or $w^B = 0$ then any $\tilde{\theta} \geq \check{\theta}$ is optimal in \cref{eq:inner_concave_alt}.
        \item Otherwise, if $\neg( w^A, w^B \leq 0)$ then the value of the program is
        \[
        \frac{R(\max\{\check{\theta}, \bar{\theta}^{\#}\})}{\max\{\check{\theta}, \bar{\theta}^{\#}\}} \left( \bar{\mathcal{W}}(\max\{\check{\theta}, \bar{\theta}^{\#}\}) + w^B \und{\theta}\right) + \int\limits_{\min\{\max\{\check{\theta}, \bar{\theta}^{\#}\}, \und{\theta}^*\}}^{\und{\theta}^*} \bar{W}(z) R'(z) dz. 
        \]
        Moreover $\bar{\theta}^{\#} \leq \und{\theta}^*$, any and all solutions take the following form: there exists
        $\tilde{\theta} \in [\max\{\check{\theta}, \und{\theta}^{\#}\}, \max\{\check{\theta}, \bar{\theta}^{\#}\}]$ such that $\und{u} = \und{\theta} R(\tilde{\theta})/\tilde{\theta}$ and 
        \begin{equation*}
        a^*(\theta) = 
        \begin{cases}
        \frac{\und{u}}{\und{\theta}} \quad &\text{on } [\und{\theta},  \min\{\tilde{\theta}, \und{\theta}^*\}] \\
        R'(\theta) &\text{on } (\tilde{\theta}, \und{\theta}^*] \\
        \text{non-increasing and bounded above by } R'(\und{\theta}^*) & \text{on } (\und{\theta}^*,\bar{\theta}^*] \\
        0 &\text{on } (\bar{\theta}^*, \bar{\theta}]
        \end{cases}
    \end{equation*}
    \end{itemize}
\end{proposition}
\begin{proof}
    Proof in \Cref{proof:concave_relax}.
\end{proof}

The solution is illustrated in \Cref{fig:tilde_c}. Note that if $\check{\theta} \leq \bar{\theta}^{\#}$ then the choice of $\tilde{\theta}$ is otherwise independent of $R$. 

\begin{figure}
    \centering
\begin{tikzpicture}[scale = 0.9, >=stealth]

  \draw[->] (-2,0) -- (9,0);
  \draw[->] (-2,-2) -- (-2,6);

  \coordinate (O)   at (-2,-2);      
  \coordinate (UL)  at (1,0);      
  \coordinate (TAN) at (3,2.5);    
  \coordinate (MAX) at (5.5,3);    
  \coordinate (END) at (8,1);      

  \draw[dashed] (O) -- ($(O)!1.7!(TAN)$) node[right] {$\theta \bar{W}(\bar{\theta}^{\#})$};

  \draw[thick,smooth]
    plot coordinates { (UL) (TAN) (MAX) (END) };

  \draw (8,0) ++(0,0.05) -- ++(0,-0.10) node[below] {$\bar{\theta}$};
  \draw (5.1,0) ++(0,0.05) -- ++(0,-0.10) node[below] {$\und{\theta}^* = \bar{\theta}^*$};
  \draw (3,0) ++(0,0.05) -- ++(0,-0.10) node[below] {$\und{\theta}^{\#}=\bar{\theta}^{\#}$};
  \draw (1,0) ++(0,0.05) -- ++(0,-0.10) node[below] {$\und{\theta}$};
  \draw (-2,0) ++(0,0.05) -- ++(0,-0.10) node[left] {$0$};
  \draw (-2,-2) ++(0,0.05) -- ++(0,-0.10) node[left] {$- w^B \und{\theta}$};

  \fill (TAN) circle(1.5pt);

  \node[above right] at (END) {$\bar{\mathcal{W}}(\theta)$};

\end{tikzpicture}
    \caption{The choice of $\tilde{\theta}$}
    \label{fig:tilde_c}
\end{figure}
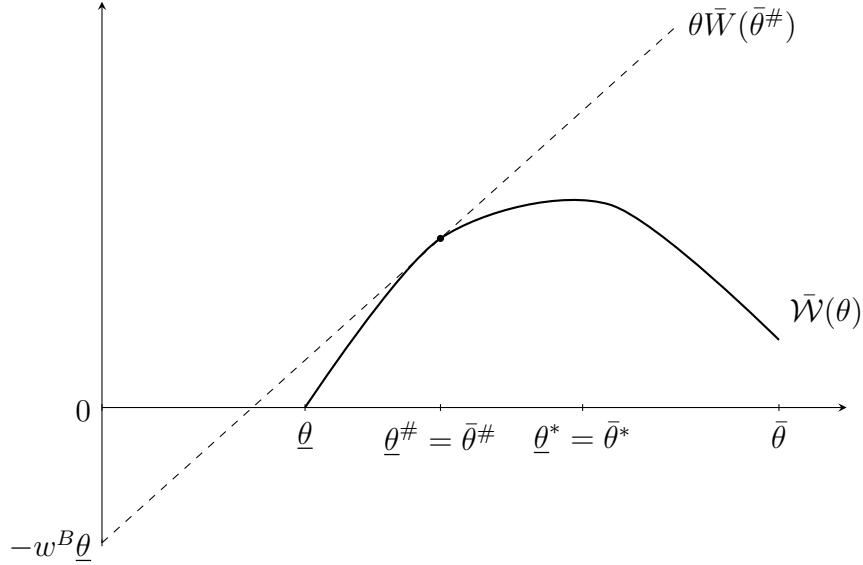

We can now turn to solving the original program in \cref{prog:alpha}. Let $\mathcal{I}$ be the (at most countable) collection of maximal open intervals $(x, y)$ within $(\und{\theta},\und{\theta}^*)$ with the property that $\mathcal{W} < \bar{\mathcal{W}}$ on $(x, y)$. Let $\mathcal{I}^c$ be the complement collection of maximal (relatively) closed intervals $[x, y]$ within $(\und{\theta},\und{\theta}^*)$ with the property that $\mathcal{W} = \bar{\mathcal{W}}$ on $[x, y]$.

\begin{theorem}\label{thm:concave}
    The value of the programs in \cref{prog:alpha} and \cref{prog:relax} are the same (which is also the value of the support function at $(w^A,w^B)$. Moreover, let $\tilde{a},\und{u}$ be a solution to the program in \cref{prog:relax}, as described in \Cref{prop:concave_relax}. Then $a^*,\und{u}$ is a solution to the program in \cref{prog:alpha}, where 
    \begin{equation*}
    a^*(\theta) = 
    \begin{cases}
        \tilde{a}(\theta) &\text{if} \ \ \theta \in [x,y] \text{ for some } [x,y] \in \mathcal{I}^c \\
        \frac{\int_x^y \tilde{a}(z)dz}{y-x} &\text{if} \ \ \theta \in (x,y) \text{ for some } (x,y) \in \mathcal{I}
    \end{cases}
\end{equation*}  
\end{theorem}
\begin{proof}
    Proof in \Cref{proof:concave}.
\end{proof}

\begin{remark}
    There exist $\theta \in [\und{\theta},\bar{\theta}]$ such that $R(\theta)/\theta > R(\und{\theta})/\und{\theta}$ if and only if $R'(\und{\theta})\und{\theta} > R(\und{\theta})$. This cannot hold if $R$ is the indirect utility from a mechanism, since in this case $R(\und{\theta}) - R'(\und{\theta})\und{\theta}$ is just the interim allocation of task $B$ to type $\und{\theta}$, which is non-negative. 
\end{remark}

\Cref{fig:welfare_concave} depicts the welfare under an optimal mechanism (dotted line) and the outside option in the concave case (solid line). In this figure there is an ironing interval $(x,y) \in \mathcal{I}$, in which the agent receives a constant allocation.

\begin{figure}
    \centering
    \begin{tikzpicture}[scale = .9]
        \draw (0,0) -- (0,6.3);
        \draw (0,0) -- (10,0);

        \draw[smooth,domain=1:10,samples=100] plot (\x,{.5 + 1.9*(\x-.5) -0.26*(\x-.5)^(1.85) + 0.006*(\x - 1)^3}) node[right] {$R$};
        
        \draw (0,0) ++(0,0.05) -- ++(0,-0.10) node[below] {$0$};
        \draw (1,0) ++(0,0.05) -- ++(0,-0.10) node[below] {$\und{\theta}$};
        \draw (10,0) ++(0,0.05) -- ++(0,-0.10) node[below] {$\bar{\theta}$};
        \draw (2,0) ++(0,0.05) -- ++(0,-0.10) node[below] {$\bar{\theta}^{\#}$};
        \draw (8,0) ++(0,0.05) -- ++(0,-0.10) node[below] {$\und{\theta}^*$};
        \draw (3,0) ++(0,0.05) -- ++(0,-0.10) node[below] {$x$};
        \draw (5.5,0) ++(0,0.05) -- ++(0,-0.10) node[below] {$y$};

        \draw (0,1) ++(0.05,0) -- ++(-0.10,0) node[left] {$\und{u} = \frac{R(\bar{\theta}^{\#})}{\bar{\theta}^{\#}}\und{\theta}$};

        \draw[very thick, dashed] (1,1) -- (2, {.5 + 1.9*(2-.5) -0.26*(2-.5)^(1.85) + 0.006*(3 - 1)^3});
        \draw[very thick,dashed,  smooth,domain=2:3,samples=100] plot (\x,{.5 + 1.9*(\x-.5) -0.26*(\x-.5)^(1.85) + 0.006*(\x - 1)^3});
        \draw[very thick, dashed] (3, {.5 + 1.9*(3-.5) -0.26*(3-.5)^(1.85) + 0.006*(3 - 1)^3}) -- (5.5,{.5 + 1.9*(5.5-.5) -0.26*(5.5-.5)^(1.85) + 0.006*(5.5 - 1)^3});
        \draw[very thick,dashed,  smooth,domain=5.5:8,samples=100] plot (\x,{.5 + 1.9*(\x-.5) -0.26*(\x-.5)^(1.85) + 0.006*(\x - 1)^3});
        \draw[very thick, dashed] (8, {.5 + 1.9*(8-.5) -0.26*(8-.5)^(1.85) + 0.006*(8 - 1)^3}) -- (10,{.5 + 1.9*(8-.5) -0.26*(8-.5)^(1.85) + 0.006*(8 - 1)^3});
    \end{tikzpicture}
    \caption{Welfare from an optimal mechanism in the concave case}
    \label{fig:welfare_concave}
\end{figure}
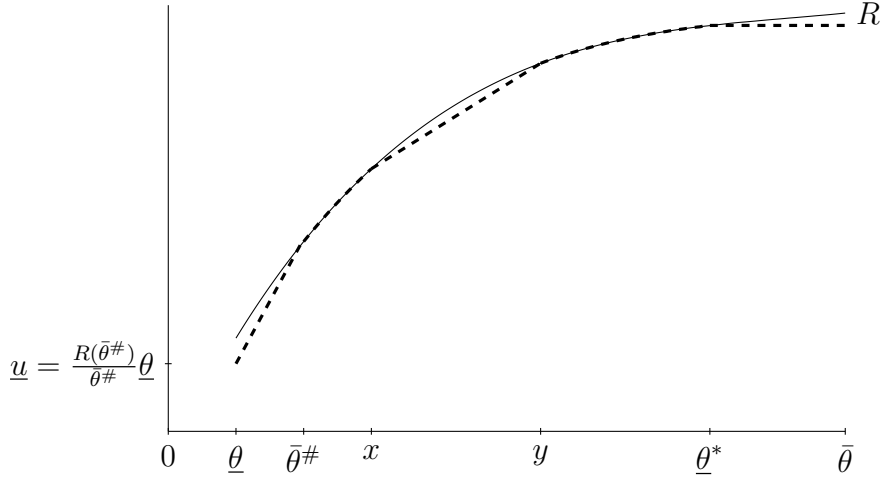

It is also instructive to consider the indirect implementation of the optimal mechanisms. For simplicity, consider the regular case where $\bar{\mathcal{W}} = \mathcal{W}$, and assume that $R(\und{\theta}) \geq R'(\und{\theta})\und{\theta}$. There is an optimal mechanism characterized by thresholds $\tilde{\theta} \in [\und{\theta}^{\#}, \bar{\theta}^{\#}]$ and $\theta^* \in [\und{\theta}^*, \bar{\theta}^*]$ such that 
\begin{equation*}
     a^*(\theta) = 
        \begin{cases}
        \frac{\und{u}}{\und{\theta}} \quad &\text{on } [\und{\theta},  \min\{\tilde{\theta}, \theta^*\}] \\
        R'(\theta) &\text{on } (\tilde{\theta}, \und{\theta}^*] \\
        0 &\text{on } (\theta^*, \bar{\theta}]
        \end{cases}
\end{equation*}
(and the optimal mechanism is unique if $\und{\theta}^{\#} =   \bar{\theta}^{\#}$ and $\und{\theta}^* = \bar{\theta}^*$). This mechanism can be implemented in the following way: We can view the reservation value $R$ as the indirect utility function of an agent who chooses from a ``reservation menu'' $\{ (R'(\theta), R(\theta) - R'(\theta)\theta )\ : \ \theta \in \Theta \}$, where $r(\theta) := (R'(\theta), R(\theta) - R'(\theta)\theta )$ is the optimal choice for type $\theta$ from this menu. The reservation menu is depicted by the blue curve in \Cref{fig:concave_indirect}. The optimal mechanism defined above is implemented by allowing the agent to choose any allocation from the reservation menu. The agent then has the option to trade for task $A$ at a rate of $\tilde{\theta}$ fewer units of task $B$ per additional unit of $A$. Additionally, the agent has the option to trade for task $B$ at a rate of $\theta^*$ additional units of task $B$ per unit of $A$ given up. The solid and dashed black lines  in \Cref{fig:concave_indirect} depict the effective ``budget set'' from which the agent can choose their allocation. 

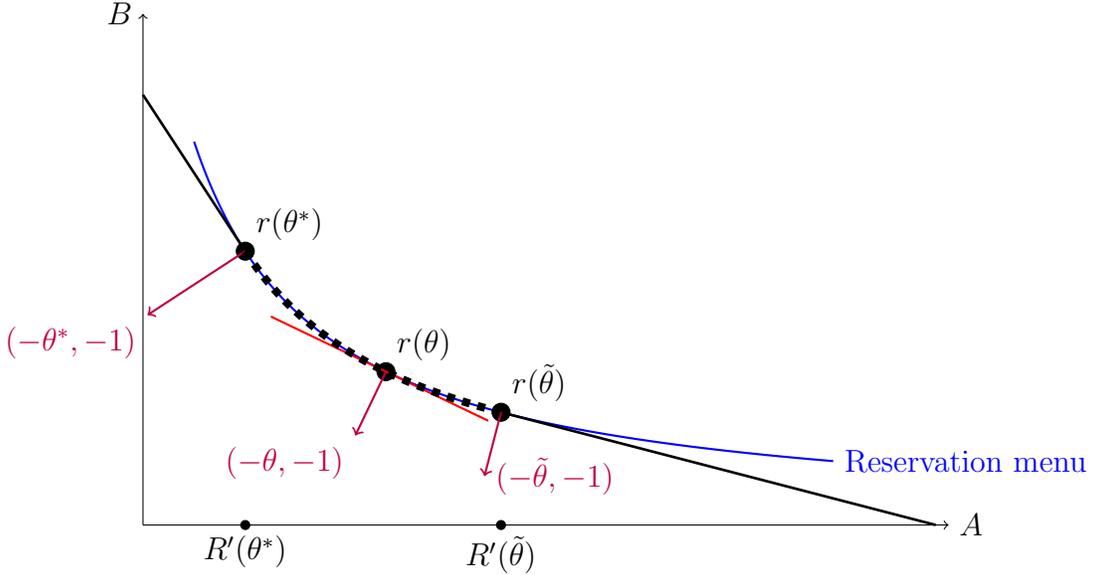
\begin{figure}[ht]
  \centering
  \begin{tikzpicture}[scale=1.7]
    \draw[->] (0,0) -- (6.3,0) node[right] {$A$};
    \draw[->] (0,0) -- (0,4) node[left] {$B$};

    \draw[blue,thick,domain=1:6,samples=200,smooth]
      plot (\x - .6, {3/\x}) node[right] {Reservation menu};

    \def\t{2.5}
    \coordinate (sq) at ({\t - 0.6}, {3/\t});
    \filldraw[black] (sq) circle (2pt) node[above right] {$r(\theta)$};

    \pgfmathsetmacro{\m}{-3/(\t*\t)}
    \draw[red,thick,domain=1:2.7,samples=2]
      plot (\x, {3/\t + \m*(\x - (\t - .6))});

    \pgfmathsetmacro{\nx}{\m}
    \pgfmathsetmacro{\ny}{-1}
    \draw[->,purple,thick] (sq) -- ++({0.5*\nx}, {0.5*\ny})
      node[below left] {$(-\theta,-1)$};

    \def\t{3.4}
    \coordinate (sq) at ({\t - 0.6}, {3/\t});
    \filldraw[black] (sq) circle (2pt) node[above right] {$r(\tilde{\theta})$};
    
    \pgfmathsetmacro{\m}{-3/(\t*\t)}
    \draw[black,line width = 1pt, domain={\t-0.6}:6.2,samples=2]
      plot (\x, {3/\t + \m*(\x - (\t - .6))});

    \pgfmathsetmacro{\nx}{\m}
    \pgfmathsetmacro{\ny}{-1}
    \draw[->,purple,thick] (sq) -- ++({0.5*\nx}, {0.5*\ny})
      node[right] {$(-\tilde{\theta},-1)$};

    \def\tt{1.4}
    \coordinate (sq) at ({\tt - 0.6}, {3/\tt});
    \filldraw[black] (sq) circle (2pt) node[above right] {$r(\theta^*)$};
    
    \pgfmathsetmacro{\m}{-3/(\tt*\tt)}
    \draw[black,line width = 1pt, domain=0:{\tt-0.6},samples=2]
      plot (\x, {3/\tt + \m*(\x - (\tt - .6))});

    \pgfmathsetmacro{\nx}{\m}
    \pgfmathsetmacro{\ny}{-1}
    \draw[->,purple,thick] (sq) -- ++({0.5*\nx}, {0.5*\ny})
      node[below left] {$(-\theta^*,-1)$};

    \draw[black, line width = 3pt, dashed, domain=\tt:\t,samples=200,smooth]
      plot (\x - .6, {3/\x});

    \filldraw[black] (\t - 0.6, 0) circle (1pt) node[below] {$R'(\tilde{\theta})$};
     \filldraw[black] (\tt - 0.6, 0) circle (1pt) node[below] {$R'(\theta^*)$};

  \end{tikzpicture}
  \caption{Indirect implementation in the concave case} \label{fig:concave_indirect}
\end{figure}

\subsubsection{Specialization and asymmetric information}\label{sec:compare_to_complete}

\Cref{thm:concave} tells us the form that optimal mechanisms will take. Already, this tells us something about the role of asymmetric information in task-allocation problems. Observe that with symmetric information, each worker generically handles only one type of task: if $\theta$ is observable, then in the inner program the firm solves
\[
\max_{a(\theta),b(\theta)} w^A a(\theta) + w^B b(\theta) \quad s.t. \quad \theta a(\theta) + b(\theta) \leq R(\theta). 
\]
A solution is the set $(a(\theta), b(\theta)) = (R(\theta)/\theta, 0)$ if $w^A/\theta \geq \max\{0, w^B\}$, $(a(\theta), b(\theta)) = (0, R(\theta))$ if $w^B \geq \max\{0, w^A/\theta\}$, and $(a(\theta), b(\theta)) = (0, 0)$ if otherwise, and this solution is unique except for knife-edge cases. Thus workers completely specialize, regardless of the degree of automation. 

In contrast, in the solution to the inner program described above, some workers with intermediate types handle both types of tasks. This is a direct result of the incentive compatibility constraints: it may be optimal for the firm to keep some intermediate-type workers at their reservation value in order to give more extreme allocations to workers with a strong preference for one or the other task. 

These observations only concern the comparison to complete information, and do not tell us how the firm allocates tasks across worker groups, or between labor and machines. To answer these questions, we must turn to the firm's outer program.

\subsubsection{A remark regarding extreme points of IC and IR mechanisms}

We have solved the program in \cref{eq:inner}, so as to characterize the support function of $\mathcal{F}$, the set of incentive-feasible aggregate allocations. However beyond simply characterizing the support function, we have identified the underlying mechanisms used to realize these aggregate task allocations. Suppose instead that we wanted to characterize the set of IC and IR mechanisms, not just aggregate allocations. To define the support function of this set, we need to maximize over all linear functionals of the mechanism, i.e., we solve
\begin{align}
  \max_{a,b \geq 0}  &\int w^A(\theta) a(\theta) + w^B(\theta) b(\theta) dF(\theta)  \\
         s.t. \quad & \theta a(\theta) + b(\theta) \leq \theta a(\theta') + b(\theta') \quad \quad \forall \ \theta,\theta' \in [\und{\theta},\bar{\theta}] \tag{IC} \\
        &\theta a(\theta) + b(\theta) \leq R(\theta) \quad \quad \forall \ \theta \in [\und{\theta},\bar{\theta}] \tag{IR}
\end{align}
for any measurable functions $w^A,w^B$. Using the envelope condition, we can write the objective and constraints as linear functions of $a$ and $\und{u}$. From there, the solution technique that we used to solve the program in \cref{prog:relax0} applies, replacing $W(\theta)$ with an arbitrary measurable function and $w^B$ with an arbitrary constant. In other words, our previous results have also characterized the extreme points and support function of the set of IC and IR mechanisms.

\section{Automation}\label{sec:automation}

Having characterized the firm's optimal mechanism, we return now to the question of automation. As discussed above, an extensive literature has offered predictions regarding the changes in workers' task allocations in response to technological changes. Our aim is to begin to understand the implications of workers' private information for these predictions.

In the current setting, an increase in automation means an increase in the allocation to the machine, $(a_m,b_m)$. The remaining tasks must be completed using labor. By \Cref{thm:outer}, the \textit{labor-cost}, $L$, as a function of the units allocated labor, $(\ell^A,\ell^B)$, is given by
\[
    L(\ell^A,\ell^B) := \sup_{\lambda^A,\lambda^B}  \left\{ \ell^A \lambda^A  + \ell^B \lambda^B  - \sum_{j\in \mathcal{J}} S_j\left(  \lambda^A - \pi_j^A  \ , \ \lambda^B - \pi_j^B \right) \right\}. 
\]
The labor-market impacts of automation are determined by how the labor cost and the allocations of each group respond to changes in $\ell^A$ and $\ell^B$. We first characterize these responses, which can be viewed as a partial-equilibrium characterization of the labor-market effects of automation. We then consider the firm's optimal-automation problem. 

\subsection{Notation and terminology}

We call a reduction in $\ell^k$ an \textit{increase in automation of task $k$}, and refer in general to a reduction in the pair $(\ell^A, \ell^B)$ as an increase in automation. Let $(\lambda_*^A(\ell),\lambda_*^B(\ell))$ be the solutions to the dual outer program given aggregate quantities $\ell$, and let $(n_j^A(\ell),n_j^B(\ell))_{j=1}^J$ be the solutions to the primal, where $(n_j^A(\ell),n_j^B(\ell))$ is the aggregate output of group $j$.

Assume throughout this section that $J \geq 3$ and groups have distinct performance, i.e., $\pi_j^k \neq \pi_{j'}^k$ for all $j,j'$ and $k \in \{A,B\}$. This assumption rules out some less interesting cases, but is not essential.\footnote{For example, if performance is identical then the solution to the outer problem is simply to set $(\lambda^A,\lambda^B) = (\pi^A,\pi^B)$ and select any profile of incentive-feasible allocations which clears the market. Thus the labor-cost function is linear (over the set of feasible $\ell$) and the results of this section are uninteresting.} To avoid cumbersome terminology, we also assume throughout this section that solutions are unique for all $\ell^A,\ell^B > 0$.\footnote{With distinct performance across groups, a sufficient condition for uniqueness is that the frontier of $\mathcal{F}_j$ be strictly convex for all $j$. Alternatively, the main results are preserved if we drop the uniqueness assumption and suitably restate the comparative statics in terms of the strong set order.} We make use of the following definitions.
\begin{itemize}
    \item A change in $\ell$ \textit{increases the output} of group $j$ in both tasks if both $n^A_j$ and $n^B_j$ weakly increase.
    \item An increase in automation of task $k$ always reduces $\lambda^k$ (the objective in the outer program has increasing differences in $\ell^k$ and $\lambda^k$). We say that an increase in automation of task $k$ has \textit{negative price spillovers} if, moreover, it reduces $\lambda_*^{-k}$, and has positive price spillovers otherwise. Under negative price spillovers, automating one task reduces the marginal cost of completing the other task with labor. The terminology is motivated by the interpretation of $(\lambda_*^A,\lambda_*^B)$ as the shadow prices on the market-clearing constraints.\footnote{Moreover, in the competitive market introduced below, $\lambda_*^k$ is precisely the equilibrium price at which completed units of task $k$ are sold, so negative price spillovers mean a reduction in the sales price of the other task.}
    \item Similarly, we call $\left(\lambda_*^A(\ell) - \pi_j^A\right) n_j^A(\ell) + \left(\lambda_*^B(\ell) - \pi_j^B\right) n_j^B(\ell)$ the \textit{surplus generated by group} $j$.
\end{itemize}

\subsection{The labor-market effects of automation} 

As the supremum over linear functions, the labor-cost function $L$ is convex in $(\ell^A,\ell^B)$. Beyond this, the properties of $L$, and thus the effects of automation, depend crucially on the size of the informational asymmetries, determined by the degree of preference heterogeneity among workers. Say that $F_j$ has \textit{low heterogeneity} if 
\begin{equation}\label{eq:dispersion}
1 - \frac{\und{\theta}_j}{\theta} \leq F_j(\theta) \leq \frac{\theta}{\bar{\theta}_j} \quad \forall \ \theta \in [\und{\theta}_j, \bar{\theta}_j]
\end{equation}
and \textit{high heterogeneity} otherwise. We say that group $j$ has \textit{low (high) information asymmetries} if $F_j$ has low (high) heterogeneity.  If we treat $\theta/\bar{\theta}_j$ and $1 - \und{\theta}_j/\theta$ as CDFs on $[\und{\theta}_j, \bar{\theta}_j]$, where the later distribution has a point mass of $1 - \und{\theta}_j/\bar{\theta}_j$ at $\bar{\theta}_j$ and the former has a point mass of $\und{\theta}_j/\bar{\theta}_j$ at $\und{\theta}$, then (\ref{eq:dispersion}) says that $F_j$ first-order stochastically dominates the former distribution, and is in turn dominated by the latter. In particular, $\frac{\bar{\theta}_j^2 + \und{\theta}_j^2}{2\bar{\theta}_j} \leq \E_F[\theta] \leq \und{\theta}_j\left(1 + \ln\left(\frac{\bar{\theta}_j}{\und{\theta}_j}\right)\right)$. Intuitively, condition (\ref{eq:dispersion}) constrains the ratio $\theta/\und{\theta}_j$. For example, if $\theta$ is uniformly distributed on $[\und{\theta},\bar{\theta}]$ then it has low heterogeneity if and only if $\bar{\theta}/\und{\theta} \leq 2$. In general, comparing the upper and lower bounds in \cref{eq:dispersion} we can see that a necessary and sufficient condition for there to exist some distribution that has support $[\und{\theta},\bar{\theta}]$ and low heterogeneity is $\bar{\theta}/\und{\theta} \leq 4$. 

\begin{theorem}\label{thm:no_disposal}
    If preference heterogeneity is small in every group then any increase in automation
    \begin{enumerate}[i.]
        \item has negative price spillovers,
        \item reduces the surplus generated by every group of workers, and 
        \item does not increase the output of any group in both tasks.
    \end{enumerate}
    Conversely, suppose that preference heterogeneity is not small for some groups. Assume, moreover, that all groups have the same incentive-feasible set, that the frontier of this set is strictly convex, and that $R_j'(\und{\theta}_j) > 0$ and $R_j(\bar{\theta}) > R_j(\und{\theta}_j)$ for all $j$.\footnote{These are technical assumptions which simplify the construction of the examples for this direction of the result, but are not essential. All groups have the same incentive-feasible set if they have the same reservation utilities and preference distributions. Note that we do not impose that $R_j$ be concave.} Then there is an increase in automation which 
    \begin{enumerate}[i.]
    \setcounter{enumi}{3}
        \item has positive price spillovers,
        \item increases the surplus generated by every group of workers, and
        \item strictly increases the output of some group in both tasks.  
    \end{enumerate}
\end{theorem}
\begin{proof}
    Proof in \Cref{proof:no_disposal}.
\end{proof}

\Cref{thm:no_disposal} reveals that information asymmetries have important implications for the effects of automation. Conclusions \textit{i.-iii.} are broadly consistent with the findings of the existing literature studying automation without information asymmetries.\footnote{For example, while the models are not directly comparable, in \cite{acemoglu2018race} and \cite{acemoglu2024automation}, under inelastic demand automation always displaces workers from the automated tasks and reduces the prices of completed tasks.} However, \Cref{thm:no_disposal} shows that these predictions are true in general only if information asymmetry is low, and can be reversed otherwise.

The driving force behind \Cref{thm:no_disposal} is that with low information asymmetries the firm treats the tasks assigned to a group of workers as substitutes: if the surplus generated by group $j$ on task $k$, $\lambda^k - \pi_j^k$, increases relative to the surplus on task $-k$, then the firm assigned the group more of task $k$ and less of task $-k$. In contrast, when asymmetries high the tasks are sometimes complements. 

Formally, call the incentive-feasible set $\mathcal{F}_j$ downward closed if $(x',y') \in \mathcal{F}_j$ and $0 \leq (x,y) \leq (x',y')$ implies $(x,y) \in \mathcal{F}_j$. In other words, the frontier of $\mathcal{F}_j$ is downward sloping: in order to allocate more of one task to group $j$, the designer must allocate less of the other task. The key step in the proof of \Cref{thm:no_disposal} is the following result. 

\begin{proposition}\label{prop:sub_dispersion}
    Assume $R_j'(\und{\theta}_j) > 0$ and $R_j(\bar{\theta}) > R_j(\und{\theta}_j)$.\footnote{It is not difficult to obtain a similar characterization without any restrictions on $R_j$, but the conditions of $F_j$ are more cumbersome to state.} Then the following are equivalent.
    \begin{enumerate}[i.]
        \item $\mathcal{F}_j$ is downward closed.
        \item $S_j$ is submodular.
        \item $F_j$ has low heterogeneity.\footnote{The equivalence between $\mathcal{F}_j$ downward closed and $S_j$ submodular is straightforward. The interesting part of \Cref{prop:sub_dispersion} is that low heterogeneity characterizes these properties.} 
    \end{enumerate}
\end{proposition}
\begin{proof}
    Proof in \Cref{proof:sub_dispersion}
\end{proof}

\Cref{prop:sub_dispersion} implies that the labor cost function is supermodular (meaning tasks are substitutes in the cost-minimization problem) when heterogeneity is low for all groups, and is used for the converse of \Cref{thm:no_disposal} to show that with high heterogeneity there are regions over which the labor-cost function is submodular. The intuition behind \Cref{prop:sub_dispersion} is that when the information rents that the firm must pay to screen workers are sufficiently high, it might be optimal in the inner program to allocate some units of task $k$ to a worker even if the weight $w^k$ is negative, provided $w^{-k} > 0$. This is because doing so helps to incentivize workers to take on task $-k$. Then even if $w^{-k}$ decreases, $w^{k}$ could increase such that the \textit{relative value} of task $-k$, i.e., $w^{-k}/|w^k|$, increases. The firm will then find it worthwhile to pay the cost of allocating more units of task $k$, in order to also allocate more of $-k$. Preference heterogeneity being low is precisely the condition which guarantees that the firm never assign tasks with a negative weight. 

Returning to \Cref{thm:no_disposal}, we can see how having an incentive-feasible set $\mathcal{F}$ with a frontier that is not downward sloping can change the effects of automation. \Cref{fig:high_heterogeneity} depicts such a situation. The solid points are allocations for each of five groups with the same incentive-feasible set, and the arrows depict the direction and magnitude of $(\lambda_*^A - \pi_j^A, \lambda_*^B - \pi_j^B)$ for each $j$.\footnote{Types here are distributed uniformly on $[0.1, 2]$ and $R(\theta) = 10\theta + 10$.} Precisely because the frontier has an upward-sloping segment, it is possible for an increase in automation to move the point $(\ell^A/J, \ell^B/J)$ closer to the frontier. This is depicted by the reduction in $\ell^A$ from the left panel to the right panel of \Cref{fig:high_heterogeneity}. Now recall that market clearing requires $(\ell^A/J, \ell^B/J)$ to be the barycenter of the groups' allocations. Because the frontier of $\mathcal{F}$ is strictly convex, the allocations across groups converge as we move towards the frontier. In particular, the allocations of some groups strictly increase (groups 4 and 5 in the figures). Moreover, for the allocations to converge, so must the angle of the normal vectors $(\lambda_*^A - \pi_j^A, \lambda_*^B - \pi_j^B)$. Because the groups have distinct performance levels, this can only happen if $\lambda_*^A \rightarrow - \infty$ and $\lambda_*^B \rightarrow + \infty$. 

\begin{figure}
    \centering
    \begin{subfigure}{.45\textwidth}
    \centering
    \includegraphics[width = 1\linewidth]{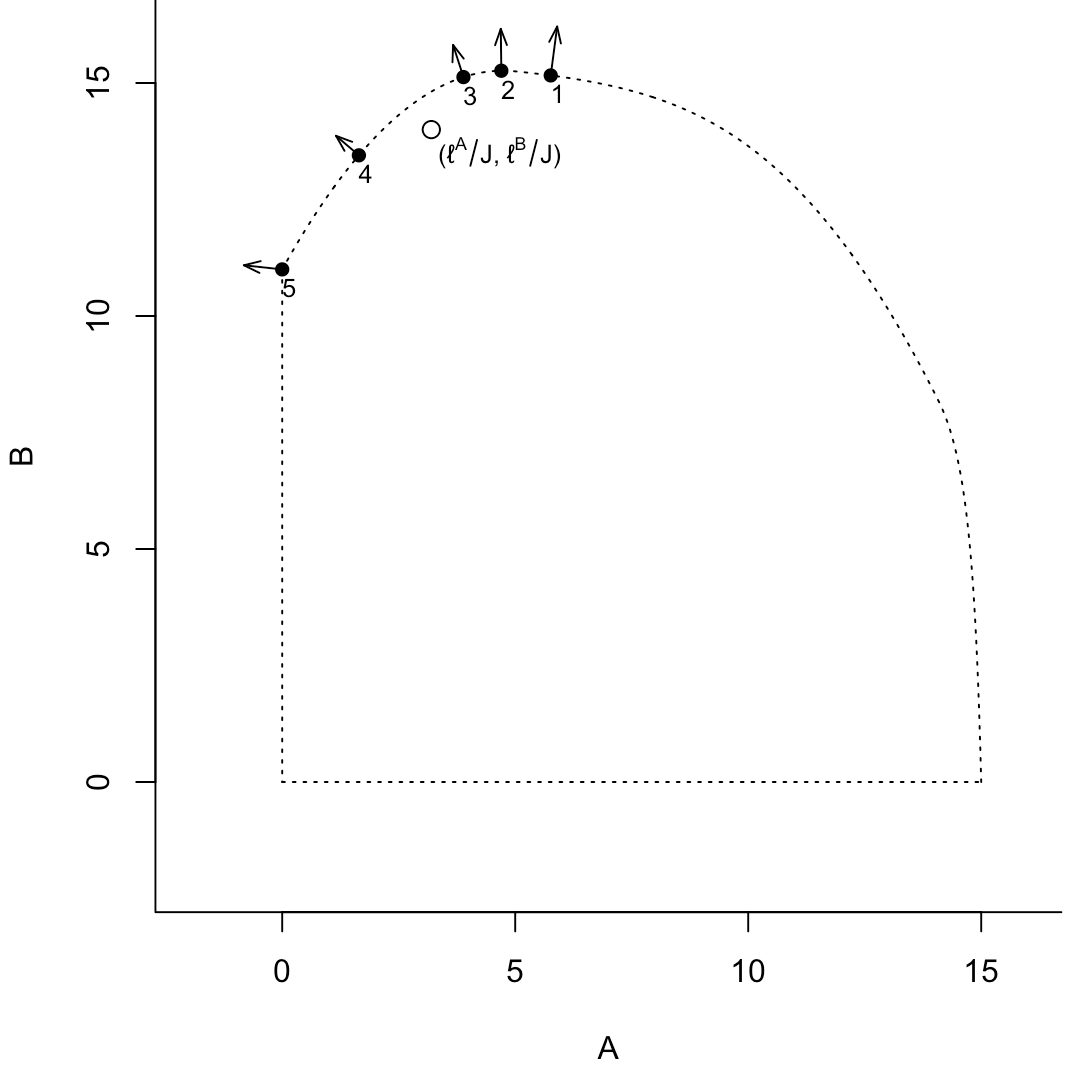}
    \end{subfigure}
    \begin{subfigure}{.45\textwidth}
    \centering
    \includegraphics[width = 1\linewidth]{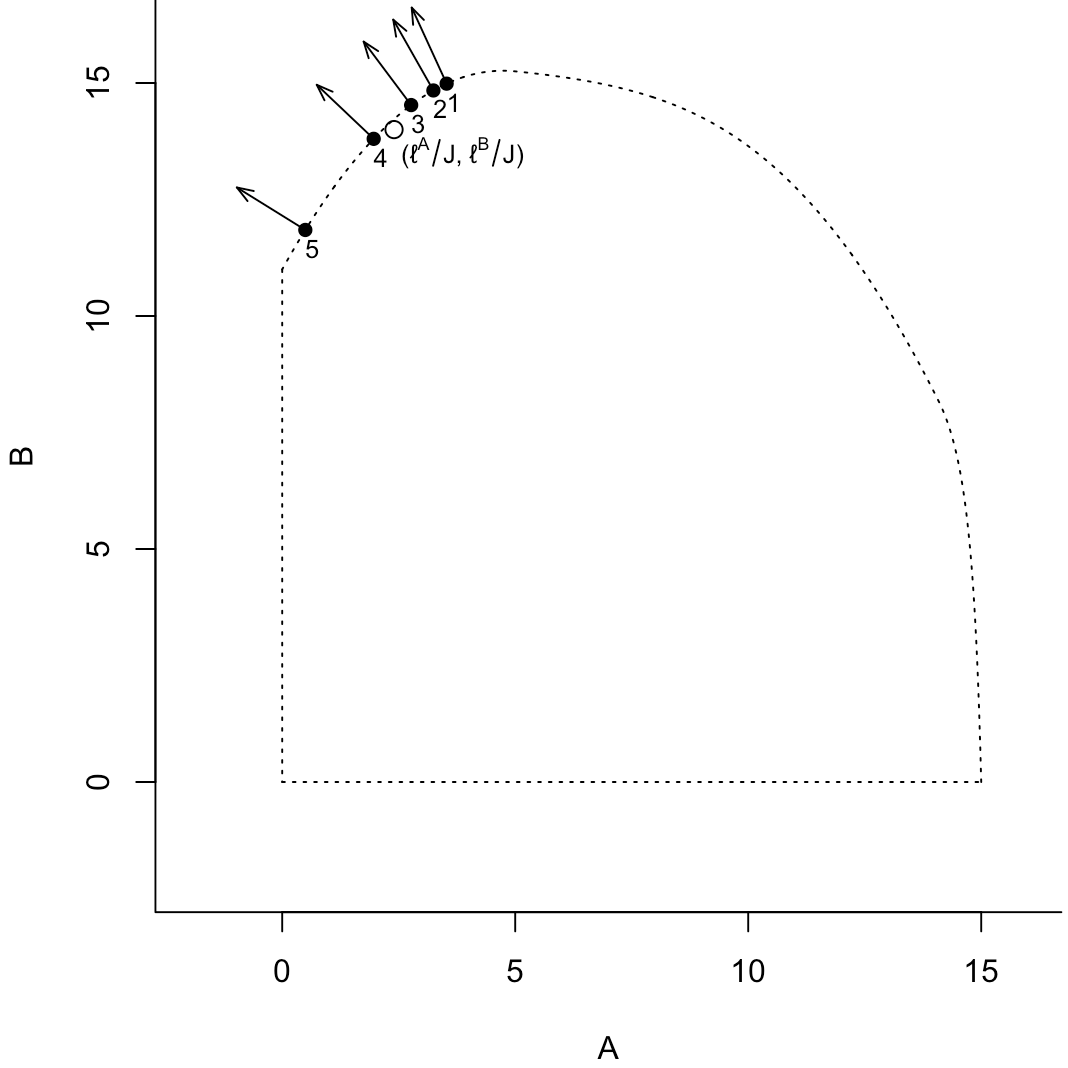}
    \end{subfigure}
    \caption{Automation with high preference  heterogeneity.}
    \label{fig:high_heterogeneity}
\end{figure}

Notice that if $\lambda^A < 0$ then the firm could actually produce higher output at a  lower cost by assigning more units of task $A$ to labor. We should expect the firm to do so if it can choose the aggregate labor allocation freely. We turn now to consider such choices.  

\subsection{Optimal labor allocations}

\Cref{thm:no_disposal} described the labor-market response to changes in the aggregate allocation of tasks to labor. We now embed the labor-allocation within the firm's broader profit-maximization program. Automation is driven by advances in machine technology, which could take the form of reductions in either the cost of computation $\gamma(\cdot)$, computation usage $c^A,c^B$, or the marginal production costs $\pi_m^A,\pi_m^B$. I refer to the changes collectively as \textit{machine-enhancing}. Our focus here is not, however, on characterizing the effects of machine-enhancing changes per se. These effects will depend on the functions $P,Y,$ and $\gamma$. Describing the implications of the shape of demand on automation, while of theoretical and practical interest, is orthogonal to the focus of the current study on the role of informational asymmetries (see for example \cite{acemoglu2018race} for coverage of demand effects).\footnote{Our characterization of $L$ is, however, a useful and necessary first step towards understanding the effects of machine-enhancing changes. Given $L$, the comparative statics of machine-enhancing changes are standard. For example, if preference heterogeneity is small, so that $L$ is submodular, and $P(Y)Y$ is supermodular, i.e the tasks are complements on the output side, then any machine-enhancing change will increase $ a_m, b_m, \mu^A,$ and $\mu^B$.} 

Our interest here is, rather, in the direct effect of automation, represented by the residual quantities $(\mu^A - a_m)$ and $(\mu^B - b_m)$ allocated to labor, on the labor market, as described by \Cref{thm:no_disposal}. Nonetheless, the firm's profit-maximization program is relevant, because it implies that some labor allocations can never be optimal for any combination of $Y,P$, and $\gamma$. This possibility is directly related to the submodularity of $S_j$ when preference heterogeneity in group $j$ is high. 

To understand why profit maximization may preclude some labor allocations, recall the examples described in \Cref{fig:high_heterogeneity}. As discussed above, for $(\ell^A/J, \ell^B/J)$ close to the upward-sloping segment of the frontier, $\lambda_*^A$ will be negative. But this means that the firm could allocate more of task $A$ to labor and produce a larger quantity at a lower cost. In other words, the labor-cost function will have non-monotonicities. Where precisely these non-monotonicities occur depends not only on the incentive-feasible set, but also on the number of groups and their performance. An example is depicted in \Cref{fig:disposal}. The gray curves are level sets of the labor-cost function in the example from \Cref{fig:high_heterogeneity}. The solid black line shows the point at which the level sets ``peak''. To the right of this line is exactly the set of points $(n^A,n^B)$ such that $\lambda_*^A(J n^A, J n^B), \lambda_*^B(J n^A, J n^B) > 0$. 

In general, allowing for the possibility that $\lambda_*^A,\lambda_*^B$ are not single-valued, define the set
\[
\mathcal{E}:= \left\{ (\ell^A/J,\ell^B/J) \in \mathcal{F} \ : \  \max\lambda_*^A(\ell^A,\ell^B) \text{ and } \max \lambda_*^B(\ell^A,\ell^B) \geq 0 \right\}
\]
Owing to \Cref{prop:eqm_set} below, we refer to the $\mathcal{E}$ as the \textit{equilibrium set}.\footnote{We maintain the symmetry assumption here for simplicity. If the incentive-feasible set differs across groups, we define the equilibrium set in terms of the aggregate labor allocation, rather than the average labor allocation:
\begin{equation}\label{eq:tilde_E}
\tilde{\mathcal{E}} := \left\{ (\ell^A,\ell^B) \in \R^2_+ \ : \ L(\ell^A,\ell^B) < +\infty,\ \max\lambda_*^A(\ell^A,\ell^B) \text{ and } \max \lambda_*^B(\ell^A,\ell^B) \geq 0 \right\}.
\end{equation}
The condition $L(\ell^A,\ell^B) < +\infty$ means that the allocation to labor is feasible. This definition is used in \Cref{sec:competitive}.
} Because non-monotonic iso-cost curves arise off of $\mathcal{E}$, and only $\mathcal{E}$ is profit-rationalizable. 

\begin{proposition}\label{prop:eqm_set}
    If $(\ell^A,\ell^B)$ is an optimal labor allocation then $(\ell^A/J, \ell^B/J) \in \mathcal{E}$. Conversely, if $(\ell^A/J, \ell^B/J) \in \mathcal{E}$ then there exists a convex computation cost $\gamma$, concave production function $Y$, and decreasing inverse demand function $P$ such that $(\ell^A, \ell^B)$ is the optimal allocation to labor.
\end{proposition}
\begin{proof}
    Proof in \Cref{proof:eqm_set}.
\end{proof}

In other words, all and only points in $\mathcal{E}$ could be the per-group optimal labor allocation. In the example in \Cref{fig:disposal}, we can see that the equilibrium set excludes the north-west corner of the incentive-feasible set. (The south-east corner is included due to the uniform distribution used in this example; more generally both corners could be excluded). To understand the jagged non-convexity of the equilibrium set, observe that this boundary is characterized by calculating the average allocation 
\[
\frac{1}{J}\sum_{j=1}^J N^*(- \pi_j^A, \lambda^B - \pi_j^B) 
\]
for varying $\lambda^B$. As we increase $\lambda^B$ the average allocation moves upwards away from the origin. The vertical jumps occur at the $\lambda^B$ to activate an additional group. Then for a fixed set of active groups, increasing $\lambda^B$ moves these groups along the upward-sloping portion of the incentive-feasible set, until the next group is activated.\footnote{As a result, the jumps on this boundary shrink as the number of groups increases, but this does not mean that the equilibrium set converges to the entire incentive-feasible set.} We can observe the following properties of the equilibrium set. 

\begin{figure}
    \centering
    \includegraphics[width = .6\linewidth]{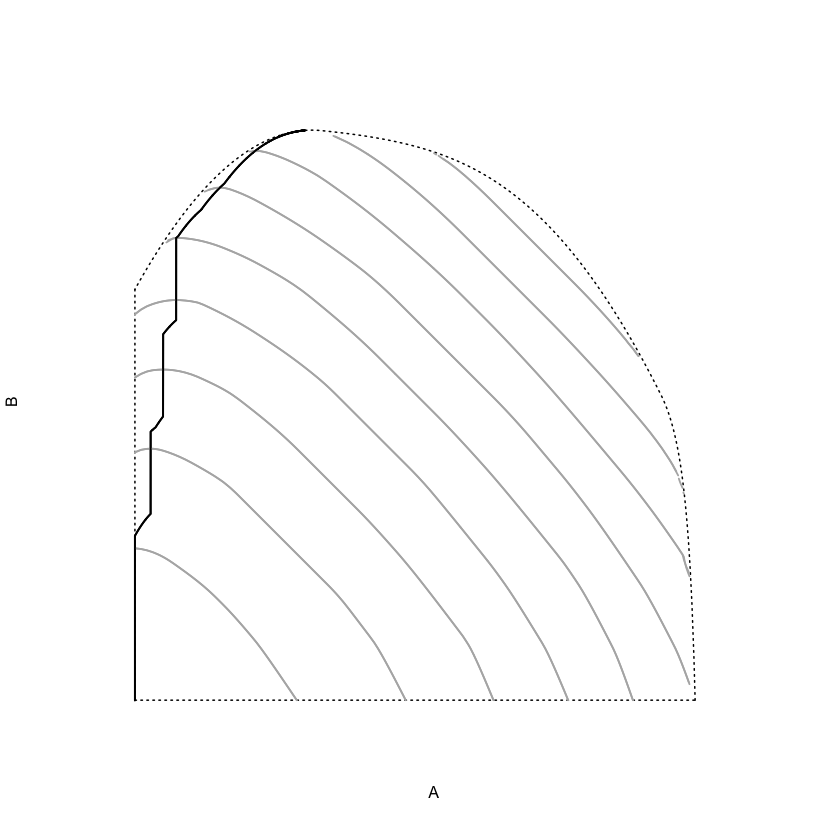}
    \caption{The labor-cost function and equilibrium set}
    \label{fig:disposal}
\end{figure}

\begin{proposition}\label{prop:eqm_properties}
   The equilibrium set is the entire incentive-feasible set (i.e., $\mathcal{E} = \mathcal{F}$) if and only if preference heterogeneity is low. Otherwise the equilibrium set
   \begin{enumerate}[i.]
       \item excludes all points on the upward-sloping segment of the frontier, except, potentially, points at which the frontier has a kink; and 
       \item includes all points on the downward-sloping segment of the frontier.
       \item If $\lambda_*^A(\ell^A,\ell^B) \geq 0$ then $\lambda_*^A(\hat{\ell}^A,\ell^B) \geq 0$ for all $\hat{\ell}^A \geq \ell^A$ such that $(\hat{\ell}^A/J,\ell^B/J) \in \mathcal{F}$. Symmetrically, if $\lambda_*^B(\ell^A,\ell^B) \geq 0$ then $\lambda_*^B(\ell^A,\hat{\ell}^B) \geq 0$ for all $\hat{\ell}^B \geq \ell^B$ such that $(\ell^A/J,\hat{\ell}^B/J) \in \mathcal{F}$.
   \end{enumerate}
\end{proposition}
\begin{proof}
    Proof in \Cref{proof:eqm_properties}.
\end{proof}

One consequence of \Cref{prop:eqm_properties} is that the ``if'' direction of \Cref{thm:no_disposal} is unchanged if we consider only optimal increases in automation. However, the restriction to $\mathcal{E}$ complicates the construction used for the converse. It is still the case that there exist examples with high preference heterogeneity in which conclusions \textit{iv.}-\textit{vi.} hold for some optimal changes in automation (the setting from \Cref{fig:high_heterogeneity} is one), however whether or not this is true for all examples with high heterogeneity remains an open question. 

\Cref{prop:eqm_properties}, along with our earlier analysis of the labor cost function, reveals two senses in which the information rents arising from asymmetric information lead to more moderate allocations to labor. First, asymmetric information induces additional convexity in the labor-cost function, relative to the case of observable worker preferences in which the labor-cost function has linear segments (corresponding to regions in which additional tasks are allocated to the same worker group). All else equal, convexity of costs tends to favor more moderate allocations. More concretely, \Cref{prop:eqm_properties} shows that as heterogeneity of worker preferences increases, we enter a situation in which certain extreme labor allocations, involving a high degree of specialization in one of the two tasks, become sub-optimal under any specification of the revenue side of the problem. 

The previous observations relate to specialization of labor on aggregate. We conclude the analysis by studying the patterns of specialization at the individual level. 

\subsection{Specialization at the individual level}

We have already observed in \Cref{sec:compare_to_complete} that asymmetric information introduces the possibility that some individual workers handle a mixed bundle of tasks, in contrast to the situation under complete information, in which workers fully specialize. We now consider how the degree of specialization responds to increases in automation. 

We first observe that in general, the individual-level effects of automation can differ across groups. Recall that in the inner program, the mechanism for group $j$ is determined by $(\lambda^A - \pi_j^A, \lambda^B - \pi_j^B)$, or equivalently, by $(\lambda^A - \pi_j^A)/|\lambda^B - \pi_j^B|$ and $\text{sign}(\lambda^B - \pi_j^B)$. Even when preference heterogeneity is low for all groups, so that any increase in automation reduces both $\lambda_*^A$ and $\lambda_*^B$ (\Cref{prop:sub_dispersion}), it could be that an increase in automation, say of task $A$, reduces the ratio $(\lambda^A - \pi_j^A)/|\lambda^B - \pi_j^B|$ for group $j$, but increases the ratio $(\lambda^A - \pi_{j'}^A)/|\lambda^B - \pi_{j'}^B|$ for group $j'$. Then the firm will substitute group $j$ from task $A$ to task $B$, but will in fact substitute group $j'$ towards the automated task $A$. \Cref{fig:high_heterogeneity} depicts an example of these opposing movements. In this example heterogeneity is high, but the same dynamic would be at play with low heterogeneity if automation moved $(\ell^A/J, \ell^B/J)$ away from a point on the frontier. 

The disparate effects of automation across groups make it difficult to draw general conclusions about the directions in which each group's mechanism will move. However, with a bit more structure we can say something about changes in the degree of individual-level specialization. 

Recall that we interpreted the reservation utility of workers as the value they were accustomed to attaining under the existing task-allocation system. If this is the case, then the reservation utility should evolve along with the mechanism. Consider the following dynamic model, in which mechanism specified for group $j$ in period $t$ defines the reservation utility for these workers in period $t+1$. That is,
\[
R^{t+1}_j(\theta) = a_j^t(\theta) \theta + b_j^t(\theta),
\]
where $(a_j^t,b_j^t)$ is the period-$t$ mechanism for group $j$, and $R_j^{t+1}$ is this group's reservation utility function. Let $\mathcal{F}_j^t$ be the incentive-feasible set for group $j$ in period $t$ (which depends on $R_j^{t-1}$). In each period, the firm solves the optimal task-allocation problem, as described up to this point. We assume here that the firm is myopic, in that it does not internalize the effect of today's mechanism on the reservation value tomorrow.\footnote{This is a strong assumption with a single firm, but it is much more palatable in the competitive-market version of the model introduced below. With a single firm, myopia of this for could arise from the short-run career concerns of managers, who want to maximize profit given the constraints they face, but who don't internalize how workers' contracts today will constrain their successor.} The question is how the mechanisms evolve in response to increasing automation over time.

Say that mechanism $(a,b)$ is \textit{more specialized} than $(\hat{a},\hat{b})$ if the following conditions hold 
\begin{enumerate}[i.]
    \item $\hat{a}(\theta) = 0$ implies $a(\theta) = 0 $, and $\hat{b}(\theta) = 0$ implies $b(\theta) = 0$,
    \item $a(\theta) >0$ and $b(\theta) >0 $ implies $\hat{a}(\theta) > 0 $ and $\hat{b}(\theta) > 0$. 
\end{enumerate}
We say that $(a,b)$ is \textit{uniformly more specialized} if in addition, $a(\theta) >0$ and $b(\theta) >0 $ implies $(\hat{a}(\theta), \hat{b}(\theta)) = (a(\theta), b(\theta))$. The first condition says that if type $\theta$ receives none of task $A$ under mechanism $(\hat{a},\hat{b})$, then this is also case under $(a,b)$. In other words, if type $\theta$ is assigned only task $k$ under $(\hat{a},\hat{b})$ then under $(a,b)$ they either assigned only task $k$ as well, or excluded entirely. The second condition says that any type that is a ``generalist'' handling some of both tasks under $(a,b)$, then they are also a generalist under $(\hat{a},\hat{b})$, and if the increase in specialization is uniform then they must receive exactly the same allocation under $(\hat{a},\hat{b})$. The more-specialized order is (very) incomplete, but when two mechanisms can be ranked in this way the comparison is unambiguous. \Cref{fig:specialization} depicts the indirect utilities in periods $t$ and $t+1$ for group $j$, showing a uniform increase in specialization.

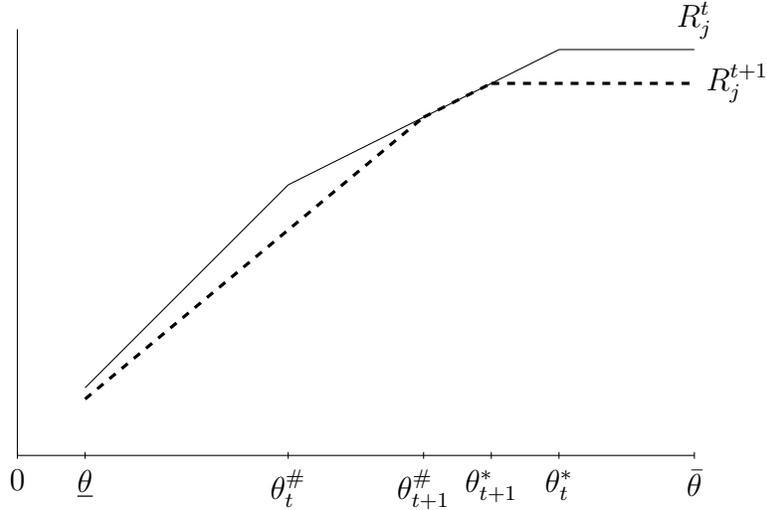
\begin{figure}
    \centering
    \begin{tikzpicture}[scale = .9]
        \draw (0,0) -- (0,6.3);
        \draw (0,0) -- (10,0);

        \draw (1,1) -- (4,4);
        \draw (4,4) -- (8,6);
        \draw (8,6) -- (10,6) node[above] {$R_j^t$};

        \draw[very thick, dashed] (1,5/6) -- (6,5);
        \draw[very thick, dashed] (6,5) -- (7,5.5); 
        \draw[very thick, dashed] (7,5.5) -- (10,5.5) node[right] {$R_j^{t+1}$};
        
        \draw (0,0) ++(0,0.05) -- ++(0,-0.10) node[below] {$0$};
        \draw (1,0) ++(0,0.05) -- ++(0,-0.10) node[below] {$\und{\theta}$};
        \draw (10,0) ++(0,0.05) -- ++(0,-0.10) node[below] {$\bar{\theta}$};
        \draw (4,0) ++(0,0.05) -- ++(0,-0.10) node[below] {$\theta_t^{\#}$};
        \draw (8,0) ++(0,0.05) -- ++(0,-0.10) node[below] {$\theta^*_t$};
        \draw (6,0) ++(0,0.05) -- ++(0,-0.10) node[below] {$\theta^{\#}_{t+1}$};
        \draw (7,0) ++(0,0.05) -- ++(0,-0.10) node[below] {$\theta^*_{t+1}$};

    \end{tikzpicture}
    \caption{Indirect utilities for a uniform increase in specialization from $R^t$ to $R^{t+1}$.}
    \label{fig:specialization}
\end{figure}

Recall that a machine-enhancing change was defined as any combination of a reduction in the cost of computation $\gamma$, computation usage $c^A,c^B$, or the marginal production costs for the machine $\pi_m^A,\pi_m^B$. 

\begin{proposition}\label{prop:specialization}
    Assume that the frontier of $\mathcal{F}^t_j$ is strictly convex for all $j$. If there is a machine-enhancing change from period $t$ to period $t+1$, then specialization increases for every group. If the frontier of $\mathcal{F}^{t+1}_j$ is strictly convex, then the increase in specialization is uniform for $j$.
\end{proposition}
\begin{proof}
    Proof in \Cref{proof:specialization}.
\end{proof}

\section{Competing firms}\label{sec:competitive}

In this section I extend the model to include multiple firms that compete to hire workers. The primary purpose of this section is to show that asymmetric information about worker preferences, rather than the firm's market power, is the main driver of the results. 

\vst
\noindent \textbf{Setup:} The production process, in which task are completed using labor or machines, is the same as in the single-firm case. The main difference is that here tasks are intermediate goods produced by competing firms, and sold to a final-goods producer. Firms must now compete for workers by offering more attractive allocations of tasks. 

\vst
\noindent\textbf{Prices and production:} Each firm looks like the monopsony producer from above: a firm employs a fixed set of workers and can complete tasks using workers or machines. For simplicity, assume that the total volume of computing resources available in the economy is $\bar{c}$. Equilibrium outcomes are unaffected by the distribution across firms of computing resources and workers.\footnote{We could alternatively assume that there is a single firm that has access to the machine, and that this firm has a convex cost of computation $\gamma(c)$, as in the single-firm model.} 

There are two main differences between the competitive and monopsony cases. First, firms in the competitive market sell completed tasks to a final goods producer at prices $\lambda^A$ and $\lambda^B$, which they take as given. More importantly, firms also compete for workers, so the problem is one of competing principals. We assume that workers' wages depend only on their group identity, and firms take these wages as given. Thus wages are not part of a firm's strategy. Rather, each firm offers a mechanism to each group of workers, consisting of functions $(a_j,b_j): \Theta_j \rightarrow \R^2_+$ specifying their task allocations if hired.\footnote{A richer model would relax the wage-taking assumption, which would allow for a more nuanced appraisal of wage dynamics. Assuming that firms take wages as given allows us to focus on competition on the task-allocation dimension. As with the single-firm model, the analysis is better suited to short-run changes, or settings with wage rigidities.} In order to hire a worker, a firm must offer them a workload that is no higher than that offered by other firms. Formally, this manifests as a participation constraint in the firm's problem, determined endogenously by the set of mechanisms offered by other firms in equilibrium. Each firm thus solves a mechanism-design problem with endogenous participation and type-dependent outside options. In principle, ensuring full participation is not without loss of optimality in such problems: a firm could potentially choose a mechanism in which some workers are not hired, so as to save on wages. Following \cite{jullien2000participation}, we can represent the workers' participation decision as part of the mechanism: let $x_j(\theta^A,\theta^B) \in \{0,1\}$ be the probability that type $(\theta^A,\theta^B)$ from group $j$ participates. 

\vst
\noindent\textbf{Equilibrium:} I focus on symmetric equilibria in which each firm offers the same mechanism to each group, and henceforth drop the ``symmetric'' qualifier. To emphasize the production side of the problem, I focus on equilibria in the intermediate goods market given a fixed demand for tasks of $(\mu^A,\mu^B)$. If there are multiple equilibria, I select those with the lowest production cost
\[
a_m \pi_m^A + b_m \pi_m^B + \sum_{j = 1}^J \E_j\left[ x_j(\theta^A,\theta^B)\left(a_j(\theta^A,\theta^B)\pi_j^A + b_j(\theta^A,\theta^B)\pi_j^B \right) \right].
\]
An equilibrium consists of output prices $\lambda_*^A$, $\lambda_*^B$, wages $(\omega_j)_{j=1}^J$, an allocation of machines to tasks $(a^*_m, b^*_m)$, and mechanisms $(x^*_j,a^*_j,b^*_j)_{j=1}^J$ for each group, such that:

\begin{enumerate}
    \item \textbf{Firms and workers optimize in the labor market.} Optimization by firms means that the mechanism offered to group $j$ must maximize the firm's profits among all feasible mechanisms. Optimization by workers means that a feasible mechanism is one which is incentive compatible and individually rational for every worker who participates.  Then for each $j \in \mathcal{J}$ the mechanism $(x^*_j,a^*_j,b^*_j)$ must solve 
    \begin{align}
   \max_{(x_j, a_j,b_j)} &\int_{\Theta_j} x_j(\theta^A,\theta^B)\left( (\lambda^A - \pi_j^A) a_j(\theta^A,\theta^B) + (\lambda^B - \pi_j^B ) b_j(\theta^A,\theta^B) - \omega_j\right)dF_j(\theta^A,\theta^B) \notag \\
    s.t. \quad &\theta^A a_j(\theta^A,\theta^B) + \theta^B b_j(\theta^A,\theta^B) \tag{IC} \\ 
        & \quad \leq \theta^A a_j(\hat{\theta}^A, \hat{\theta}^B) + \theta^B b_j(\hat{\theta}^A,\hat{\theta}^B) \quad\quad\quad\quad \ \text{if} \ \ x_j(\hat{\theta}^A,\hat{\theta}^B) = 1 \notag \\
        &\theta^A a_j(\theta^A,\theta^B) + \theta^B b_j(\theta^A,\theta^B)  \leq r_j(\theta^A,\theta^B) \quad \text{if} \ \ x_j(\theta^A,\theta^B) = 1 \tag{IR} \\
        & \omega_j - \theta^A a_j(\theta^A,\theta^B) - \theta^B b_j(\theta^A,\theta^B)  \tag{C} \\
        & \quad \geq \max\left\{ x^*_j(\theta^A,\theta^B) \left( \omega_j - \theta^A a^*_j(\theta^A,\theta^B) - \theta^B b^*_j(\theta^A,\theta^B) \right), \ U \right\}\  \text{if} \  x_j(\theta^A,\theta^B) =1\notag 
    \end{align}
    The (IC) constraint says that workers report truthfully (we can set  $a_j(\theta^A,\theta^B) = b_j(\theta^A,\theta^B) = 0$ if $x_j(\theta^A,\theta^B) = 0$). The (IR) constraint represents the internal constraints on the workload for workers who are hired, as in the monopsony case. We distinguish this from the competition constraint (C), which represents the worker's participation decision given the equilibrium contracts $(a^*_j,b^*_j)$ and the payoff from being unemployed, denoted by $U$.  
 
    \item \textbf{Machine use is optimal}. The allocation of machines to tasks, $(a_m,b_m)$, solves
    \[
    \max_{(a_m,b_m)} (\lambda^A - \pi_m^A)a_m + (\lambda^B - \pi_m^B)b_m \ \ s.t. \ \ a_m c^A + b_m c^B \leq \bar{c}.
    \]
    \item \textbf{The intermediate-goods market clears.} 
    \[
    a_m + \sum_{j \in \mathcal{J}} \int_{\Theta_j} a_j(\theta^A,\theta^B) dF_j(\theta^A,\theta^B) \geq \mu^A \quad \text{and} \quad b_m + \sum_{j \in \mathcal{J}} \int_{\Theta_j} b_j(\theta^A,\theta^B) dF_j(\theta^A,\theta^B) \geq \mu^B
    \]
    Note that we allow for free disposal of intermediate goods. 
\end{enumerate}

We say that the problem is feasible if there exists a set of mechanisms $(x_j,a_j,b_j)_{j \in \mathcal{J}}$ which satisfy (IC), (IR), and market clearing. Thus feasibility does not require firm optimality or impose the competition or exclusion conditions, it simply means that it is possible to produce the desired output given the workers' outside options and the distribution of preferences. In other words, the output levels could be produced by a monopolistic firm with access to the same resources. This is clearly a necessary condition for equilibrium existence. 

To verify existence, and characterize minimum-cost equilibria, we first solve for a candidate equilibrium in which  $x^*_j(\theta^A,\theta^B) = 1$ for all $j \in \mathcal{J}$ and $(\theta^A,\theta^B) \in \Theta_j$. This is in fact an equilibrium if and only if it gives each worker a payoff which is better than the outside option of unemployment. For simplicity, we focus on this case.

\begin{theorem}\label{thm:competitive}
    Assume the problem is feasible and unemployment is sufficiently unattractive.\footnote{Precisely how low the value of unemployment, $U$, should be is specified in the proof of \Cref{thm:competitive}.} Then there exists a minimum cost equilibrium that features full participation. In particular $a^*_m,b^*_m, (a^*_j,b^*_j)_{j=1}^J$ are the solutions to the monopsony program in \cref{prog:1}, with the market-clearing conditions replaced by weak inequalities, and the prices $\lambda^A_*,\lambda^B_* \geq 0$ are the multipliers on these constraints. 
\end{theorem}
\begin{proof}
    Proof in \Cref{proof:competitive}
\end{proof}

In other words, the outcomes with competing firms are equivalent to those when the market is controlled by a single firm. We replace the market-clearing conditions with weak inequalities because we assume free disposal in the competitive market. This implies that the prices must be non-negative, and the aggregate labor allocation will belong to the equilibrium set $\tilde{\mathcal{E}}$ (as defined for the asymmetric case in \cref{eq:tilde_E}).

\section{Conclusion}\label{sec:conclusion}

This paper shows how workers’ private information about their preferences over tasks shapes optimal task allocation. The main lesson is that asymmetric information matters for understanding the labor-market implications of technological change: there is a sharp divide between patterns of automation in the high- and low-asymmetry cases. The simple model studied here omits a number of important real-world factors, such as strategic wage setting by firms, demand effects, and production using more than two tasks. Nonetheless, the results suggest that models which incorporate these factors without accounting for workers' private information may miss important details about the interaction between automation and task allocation. Further work on modeling asymmetric information in a richer task-allocation framework would be valuable for generating empirically testable predictions. Empirical evidence documenting the degree of preference heterogeneity among workers would also be useful.

\section*{Appendix}
\appendix

\section{The inner program for general \texorpdfstring{$R$}{}}\label{sec:general}

\subsection{The convex case}

Suppose now that $R$ is convex and non-decreasing. We solve the program by mapping it back to the concave case. Recall the IR constraint
\[
\und{u} + \int_{\und{\theta}}^{\theta} a(z)dz \leq R(\theta).
\]
The left hand side is a concave function since $a$ must be non-increasing. Thus by the separating hyperplane theorem, for any feasible $a,\und{u}$ there exists a point $\theta' \in [\und{\theta},\bar{\theta}]$ such that
\[
\und{u} + \int_{\und{\theta}}^\theta a(z)dz \leq R(\theta') + (\theta - \theta')R'(\theta') 
\]
for all $\theta$. This is depicted in \Cref{fig:welfare_convex}. In other words, we can replace $R$ with one of its supporting hyperplanes and and obtain the same value of the program. The question is which supporting hyperplane. To answer this, we need to consider replacing $R$ with $\hat{R}_{\theta'}(\theta) := R(\theta') + (\theta - \theta')R'(\theta')$, solve for the optimal mechanism given this constraint, and then optimize over $\theta'$. This problem is greatly simplified by the solution the concave case described by \Cref{prop:concave_relax} and \Cref{thm:concave}. First, let $\hat{t} = \max\{ R(\theta)/\theta \ : \ \theta \in [\und{\theta},\bar{\theta}] \}$ and let $\hat{\theta} = \max\{\theta \in [\und{\theta},\bar{\theta}] : R(\theta)/\theta = \hat{t} \}$.

\begin{lemma}\label{lem:cprime_upperbound}
    There is an optimal $\theta' \leq \hat{\theta}$. 
\end{lemma}
\begin{proof}
    If $\theta' > \hat{\theta}$ then by convexity of $R$ we have $R'(\theta') > \hat{t}$. But then $\hat{R}_{\theta'}(\und{\theta})/\und{\theta} \leq \hat{R}_{\theta'}(\theta)/\theta$ for all $\theta \in [\und{\theta},\bar{\theta}]$. This implies that in the solution to the program under $\hat{R}_{\theta'}$ the constraint $\und{u} \leq R_{\theta'}(\und{\theta})$ will bind. In other words, we are in the case with $\check{\theta} > \bar{\theta}^{\#}$ from \Cref{prop:concave_relax}. Let $(a,\und{u})$ be the solution. Then if we reduce $\theta'$ to $\hat{\theta}$, $(a,\und{u})$ remains feasible. 
\end{proof}

\begin{figure}
    \centering
    \begin{tikzpicture}
        \draw (0,0) -- (0,5);
        \draw (0,0) -- (10,0);

        \draw[smooth,domain=1:10,samples=100] plot (\x,{1.5 + (0.2*(\x-0.7))^(2)}) node[right] {$R$};
        
        \draw (0,0) ++(0,0.05) -- ++(0,-0.10) node[below] {$0$};
        \draw (1,0) ++(0,0.05) -- ++(0,-0.10) node[below] {$\und{\theta}$};
        \draw (10,0) ++(0,0.05) -- ++(0,-0.10) node[below] {$\bar{\theta}$};
        \draw (3,0) ++(0,0.05) -- ++(0,-0.10) node[below] {$\bar{\theta}^{\#}$};
        \draw (7,0) ++(0,0.05) -- ++(0,-0.10) node[below] {$\und{\theta}^*$};
        \draw (4,0) ++(0,0.05) -- ++(0,-0.10) node[below] {$\theta'$};

        \draw[smooth,dashed,domain=1:10,samples=100] plot (\x,{1.5 + (0.2*(4-0.7))^(2) +(\x - 4)*(0.4)*(0.2*(4 - 0.7))  - (0.15*(\x-4))^(2)})  node[right] {$\und{u} + \int_{\und{\theta}}^\theta a(z)dz$};

        \draw[thick, smooth,dotted,domain=1:10,samples=100] plot (\x,{1.5 + (0.2*(4-0.7))^(2) +(\x - 4)*(0.4)*(0.2*(4 - 0.7)) });

    \end{tikzpicture}
    \caption{Welfare from a feasible mechanism in the convex case}
    \label{fig:welfare_convex}
\end{figure}
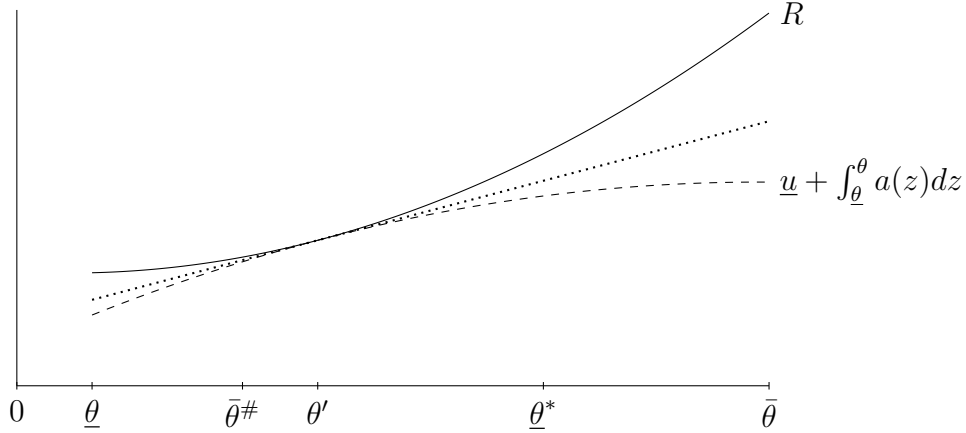

If $w^A,w^B \leq 0$, so that $\mathcal{W}(\und{\theta}^*) + w^B\und{\theta} \leq 0$, then any choice of $\theta'$ is optimal, and the solution is unchanged from the concave case. Otherwise, from \Cref{lem:cprime_upperbound} we can see that we will always choose $\theta'$ such that the solution to the relaxed problem is in the $\check{\theta} \leq \bar{\theta}^{\#}$ case from \Cref{prop:concave_relax}. Because $\bar{\mathcal{W}}(\bar{\theta}^{\#}) = \mathcal{W}(\bar{\theta}^{\#})$ and $\bar{\mathcal{W}}(\und{\theta}^*) = \mathcal{W}(\und{\theta}^*)$, the value of the program is 
    \begin{align}
    &\frac{\hat{R}_{\theta'}(\bar{\theta}^{\#})}{\bar{\theta}^{\#}} \bar{\mathcal{W}}(\bar{\theta}^{\#}) + \int_{\bar{\theta}^{\#}}^{\und{\theta}^*} \bar{W}(z)\hat{R}_{\theta'}'(z)dz + w^B\und{\theta} \frac{\hat{R}_{\theta'}(\bar{\theta}^{\#})}{\bar{\theta}^{\#}} \notag \\
    &= \frac{\hat{R}_{\theta'}(\bar{\theta}^{\#})}{\bar{\theta}^{\#}} \mathcal{W}(\bar{\theta}^{\#}) +  \left(\mathcal{W}(\und{\theta}^*) - \mathcal{W}(\bar{\theta}^{\#})\right)R'(\theta') + w^B\und{\theta} \frac{\hat{R}_{\theta'}(\bar{\theta}^{\#})}{\bar{\theta}^{\#}} \label{eq:theta_prime}
    \end{align}
which we need to maximize over $\theta'$. 

\begin{theorem}\label{thm:convex}
    If $R$ is convex then the solutions to the inner program are given by \Cref{thm:concave} with $R$ replaced by the affine function $\hat{R}_{\theta'}$, where $\theta'$ is arbitrary if $w^A,w^B \leq 0$, and otherwise
\begin{equation*}
\theta' = 
    \begin{cases}
        \hat{\theta} &\text{if } \bar{\theta}^{\#} \frac{\mathcal{W}(\und{\theta}^*) + w^B\und{\theta}}{\mathcal{W}(\bar{\theta}^{\#}) + w^B\und{\theta}} \geq \hat{\theta} \\
        \bar{\theta}^{\#} \frac{\mathcal{W}(\und{\theta}^*) + w^B\und{\theta}}{\mathcal{W}(\bar{\theta}^{\#}) + w^B\und{\theta}} & \text{otherwise}
    \end{cases}
\end{equation*} 
\end{theorem}
\begin{proof}
    Proof in \Cref{proof:convex}.
\end{proof}

A few pertinent observations: First $\theta' = \und{\theta} + \frac{\mathcal{W}(\und{\theta}^*)}{w^B}$ if  $\bar{\theta}^{\#} = \und{\theta}$. Second when $\bar{\theta}^{\#}$ is interior, so that $\bar{\mathcal{W}}(\bar{\theta}^{\#}) = \bar{W}(\bar{\theta}^{\#})\bar{\theta}^{\#} - \und{\theta}w^B$, we have  
\begin{equation*}
\bar{\theta}^{\#} \frac{\bar{\mathcal{W}}(\und{\theta}^*) + w^B\und{\theta}}{\bar{\mathcal{W}}(\bar{\theta}^{\#}) + w^B\und{\theta}} =  \frac{\bar{\mathcal{W}}(\und{\theta}^*) + w^B\und{\theta}}{\bar{W}(\bar{\theta}^{\#})} 
\end{equation*}
In particular, this implies that  $\bar{\theta}^{\#} \leq \theta' \leq \und{\theta}^*$. 

\subsection{The general case}

In the general case, $[\und{\theta}, \bar{\theta}]$ is partitioned into intervals on which $R$ alternates between concave and convex. Let $\mathcal{X}^{cc} = \{X^{cc}_1,\dots,X^{cc}_n\}$ and $\mathcal{X}^{cv} = \{X^{cv}_1,\dots,X^{cv}_n\}$ be the set of concave and convex intervals respectively, with the normalization that the convex intervals are closed in $[\und{\theta},\bar{\theta}]$. 

As in the convex case, we will find the solution by replacing $R$ with some concave lower-bound, which we denote here by $\hat{R}$. While it is difficult to explicitly characterize the optimal $\hat{R}$ in general, we can identify some properties that it must have. 

For any $s \geq 0$ and $i \geq 1$, define 
\[
\tilde{R}_{i,s}(\theta) := \max\{a + s c : a + s x \leq R(x) \ \forall \ x \in X_i^{cv} \}.
\]
So if $s \in R'(X_i^{cv})$ then $\tilde{R}_{i,s}$ is a tangent hyperplane to $R$ at some $\theta \in X_i^{cv}$. Otherwise, if $s > \max R'(X_i^{cv})$ then $\tilde{R}_{i,s}(\theta) = R(\max X_i^{cv}) + s (\theta - \max X_i^{cv})$, and if $s < \min R'(X_i^{cv})$ then  $\tilde{R}_{i,s}(\theta) = R(\min X_i^{cv}) + s (\theta - \min X_i^{cv})$). 

\begin{proposition}\label{prop:general_hatR}
    There is an optimal effective upper bound, $\hat{R}$, with the following properties
    \begin{enumerate}
        \item $\hat{R}(\theta)$ is constant above $\und{\theta}^*$.
        \item For each interval $X^{cv}_i \in \mathcal{X}^{cv}$ there is a slope $s_i \geq 0$ such that $\hat{R}$ is the lower envelope of $\{ R, \tilde{R}_{1,s_1}, \dots, \tilde{R}_{n,s_n}\}$.
        \item $s_i \leq s_j$ for $i > j$. 
        \item $s_i \leq \min\{R(\und{\theta})/\und{\theta}), R(\bar{\theta}^{\#})/\bar{\theta}^{\#} \}$.
    \end{enumerate}
\end{proposition}
\begin{proof}
    Proof in \Cref{proof:general_hatR}.
\end{proof}
Properties 1 and 4 depend on the structure of the solution for the concave case, identified in \Cref{prop:concave_relax}. In contrast, properties 2 and 3 hold for any pointwise-undominated $\hat{R}$.

\section{Omitted proofs}\label{sec:omitted_proofs}

\subsection{Proof of \texorpdfstring{\Cref{lem:effective_R}}{}}\label{proof:effective_R}

\begin{proof}
    IC and IR are satisfied if and only if
    \begin{align*}
        \theta^A a_j(\theta^A,\theta^B) + \theta^B b_j(\theta^A,\theta^B) &= \min_{(\hat{\theta}^A,\hat{\theta}^B) \in \Theta_j} \theta^A a_j(\hat{\theta}^A,\hat{\theta}^B) + \theta^B b_j(\hat{\theta}^A,\hat{\theta}^B) \\
        & = \theta^B \min_{(\hat{\theta}^A,\hat{\theta}^B) \in \Theta_j} \frac{\theta^A}{\theta^B} a_j(\hat{\theta}^A,\hat{\theta}^B) + b_j(\hat{\theta}^A,\hat{\theta}^B) \\
        & \leq r_j(\theta^A,\theta^B)
    \end{align*}
    where the first line is the definition of IC and the last line is the definition of IR. Notice that if $\tilde{\theta}^A/ \tilde{\theta}^B = \theta^A/\theta^B$ then $\min_{(\hat{\theta}^A,\hat{\theta}^B) \in \Theta_j} \frac{\theta^A}{\theta^B} a_j(\hat{\theta}^A,\hat{\theta}^B) + b_j(\hat{\theta}^A,\hat{\theta}^B) = \min_{(\hat{\theta}^A,\hat{\theta}^B) \in \Theta_j} \frac{\tilde{\theta}^A}{\tilde{\theta}^B} a_j(\hat{\theta}^A,\hat{\theta}^B) + b_j(\hat{\theta}^A,\hat{\theta}^B).$ Then 
    \begin{align*}
    \frac{\theta^A}{\theta^B} a_j(\theta^A,\theta^B) + b_j(\theta^A,\theta^B) &= \frac{\tilde\theta^A}{\tilde\theta^B} a_j(\tilde\theta^A,\tilde\theta^B) + b_j(\tilde\theta^A,\tilde\theta^B) \\
    &\leq \frac{1}{\tilde\theta^B} r_j(\tilde{\theta}^A,\tilde{\theta}^B),
    \end{align*}
    for all $(\theta^A,\theta^B),(\tilde{\theta}^A,\tilde{\theta}^B)$ such that $\tilde{\theta}^A/ \tilde{\theta}^B = \theta^A/\theta^B$, as desired.  
\end{proof}

\subsection{Alternative proof of \texorpdfstring{\Cref{lem:fixed_u}}{}}\label{proof:altfixed_u}

\begin{proof}
Without loss of optimality we restrict attention to upper-semicontinuous $a$. The non-negativity constraint and $a$ non-increasing implies $a(\theta) \leq \frac{\und{u}}{\und{\theta}}$ for all $\theta$. The extreme points of the set of upper-semicontinuous non-decreasing functions from $[\und{\theta},\bar{\theta}] \rightarrow [0,\frac{\und{u}}{\und{\theta}}]$ are step functions of the form $\frac{\und{u}}{\und{\theta}} \mathbbm{1}_{\theta \leq x}$. Therefore by Choquet's theorem (e.g., \cite{phelps2002lectures}) we can write a (upper-semicontinuous, non-increasing) function $a$ as 
\[
a(\theta) = \int_{\und{\theta}}^{\bar{\theta}} \frac{\und{u}}{\und{\theta}} \mathbbm{1}_{\theta \leq x} dG(x)
\]
where $G$ is a cdf on $[\und{\theta},\bar{\theta}]$. Then the firm's objective is 
\begin{align*}
    \int_{\und{\theta}}^{\bar{\theta}} a(\theta) \bar{W}(\theta)d\theta &= \int_{\und{\theta}}^{\bar{\theta}} \left( \int_{\und{\theta}}^{\bar{\theta}} \frac{\und{u}}{\und{\theta}} \mathbbm{1}_{\theta \leq x} dG(x) \right) \bar{W}(\theta)d\theta \\
    & = \frac{\und{u}}{\und{\theta}} \int_{\und{\theta}}^{\bar{\theta}} \bar{\mathcal{W}}(x)dG(x). 
\end{align*}
The non-negativity and non-increasing constraints are implied by the Choquet representation of $a$, so we need only impose the IR constraint $\und{u} + \int_{\und{\theta}}^\theta a(z)dz \leq R(\theta)$. Using the Choquet representation of $a$ this constraint becomes
\begin{align*}
R(\theta) &\geq \und{u} + \frac{\und{u}}{\und{\theta}} \int_{\und{\theta}}^{\bar{\theta}} \left(\min \{x ,\theta \} - \und{\theta} \right) dG(x) \\
&= \frac{\und{u}}{\und{\theta}} \left( \int_{\und{\theta}}^{\theta} x \ dG(x) + \theta(1- G(\theta)) \right)
\end{align*}
The firm's problem is then to choose the distribution $G$ to maximize the objective $\frac{\und{u}}{\und{\theta}} \int_{\und{\theta}}^{\bar{\theta}} \bar{\mathcal{W}}(x)dG(x)$, subject to this constraint. 

Observe that a first-order-stochastic-dominance (FOSD) upward (resp. downward) shift in $G$ tightens (resp. relaxes) the IR constraint for all $\theta$. Since $\bar{\mathcal{W}}$ is strictly decreasing above $\bar{\theta}^*$, this implies that no optimal $G$ can place mass above $\bar{\theta}^*$: shifting mass to $\bar{\theta}^*$ relaxes the IR constraint and improves the objective. 
On the domain $[\und{\theta},\bar{\theta}^*]$ we have $\bar{\mathcal{W}}$ increasing. Note that $\tilde{\theta}$ as defined in \Cref{lem:fixed_u} is the unique point such that the IR constraint binds if $G$ jumps from $0$ to some positive value at $\tilde{\theta}$, i.e., such that
\[
R(\theta) = \frac{\und{u}}{\und{\theta}} \left( \theta G(\theta) + \theta(1 - G(\theta)) \right) = \frac{\und{u}}{\und{\theta}} \theta. 
\]
Then the following distribution first-order stochastically dominates all others on this domain that satisfy the IR constraint: $G(\theta) = 0$ if $\theta \leq \tilde{\theta}$ and $R(\theta) = \frac{\und{u}}{\und{\theta}} \left( \int_{\und{\theta}}^{\theta} x \ dG(x) + \theta(1- G(\theta)) \right)$ for $\theta \in [\tilde{\theta}, \bar{\theta}^*]$. This is therefore a solution. Moreover, the only other distributions which achieve the same value are those which coincide with $G$ outside of $[\und{\theta}^*,\bar{\theta}^*]$, and all such distributions are optimal. These are precisely the solution identified in \Cref{lem:fixed_u}. 
\end{proof}

\subsection{Proof of \texorpdfstring{\Cref{prop:concave_relax}}{}}\label{proof:concave_relax}

\begin{proof}
First, observe that $\mathcal{W}(\und{\theta}^*) + w^B \und{\theta} \leq (<) 0$ if and only if $w^A,w^B \leq (<) 0$.

There are two cases to consider. If $\theta < \und{\theta}^*$ then it is optimal to increase $\theta$ as long as\footnote{$R$ is differentiable almost everywhere by assumption, and we use the left derivative of $R$ at points of non-differentiability.} 
\begin{align*}
&\frac{R'(\theta)\theta - R(\theta)}{\theta^2}\left(\bar{\mathcal{W}}(\theta) + w^B \und{\theta} \right) + \frac{R(\theta)}{\theta} \bar{W}(\theta)  - \bar{W}(\theta) R'(\theta) > 0 \\
& \Leftrightarrow R'(\theta)\left(\bar{\mathcal{W}}(\theta) + w^B \und{\theta} - \theta \bar{W}(\theta) \right) - \frac{R(\theta)}{\theta}\left(\bar{\mathcal{W}}(\theta) + w^B \und{\theta} - \theta \bar{W}(\theta) \right) > 0 \\
& \Leftrightarrow \bar{\mathcal{W}}(\theta) + w^B \und{\theta} - \theta \bar{W}(\theta) < 0 
\end{align*}
where the final equivalence follows from the fact that $R(\theta)/\theta > R'(\theta)$ for all $\theta \geq \check{\theta}$. If $\theta \geq \und{\theta}^*$ then it is optimal to increase $\theta$ as long as
\begin{align*}
    &\frac{R'(\theta)\theta - R(\theta)}{\theta^2}\left(\bar{\mathcal{W}}(\und{\theta}^*) + w^B \und{\theta} \right) > 0 \\
    & \Leftrightarrow \bar{\mathcal{W}}(\und{\theta}^*) + w^B \und{\theta} < 0
\end{align*}
where the equivalence again follows from the fact that $R'(\theta) < R(\theta)/\theta$ above $\check{\theta}$. Notice that if $\bar{\mathcal{W}}(\und{\theta}^*) + w^B \und{\theta} < 0$ then $\bar{\mathcal{W}}(\theta) + w^B \und{\theta} - \theta \bar{W}(\theta) < 0$ for all $\theta \in [\check{\theta},\und{\theta}^*]$. 

    The first case of \Cref{prop:concave_relax} follows from the preceding observations, given that $\max \bar{\mathcal{W}}(\theta) = \max \mathcal{W}(\theta)$. Then it is optimal to increase $\theta \rightarrow \infty$ in the program in \cref{eq:inner_concave_alt}.

    For the second case of \Cref{prop:concave_relax}, note that $\theta > \bar{\theta}^{\#}$ implies $\bar{\mathcal{W}}(\theta) + w^B\und{\theta} > \theta \bar{W}(\theta)$, so it is optimal to decrease $\theta$ in \cref{eq:inner_concave_alt}. Then the solution is determined by whether or not the constraint $\theta \geq \check{\theta}$ binds. 
\end{proof}

\subsection{Proof of \texorpdfstring{\Cref{thm:concave}}{}}\label{proof:concave}

\begin{proof}
    Starting from \Cref{prop:concave_relax}, the proof of \Cref{thm:concave} is nearly identical to the standard ironing argument of \citep{myerson1981optimal}. That is, the value of the program in \cref{prog:relax} is an upper bound on the original program. On an ironing interval $(x,y)$, the constancy of $\bar{\mathcal{W}}$ implies $\int_{x}^y a(\theta)W(\theta)d\theta =\int_{x}^y \tilde{a}(\theta)\bar{W}(\theta)d\theta$, so $a$ achieves the same value as $\tilde{a}$ in the program in \cref{prog:relax}. Moreover $a$ is non-increasing, and $b(\theta) \geq 0$ continues to hold because ironing weakly reduces $a(\und{\theta})$ while leaving $\und{u}$ unchanged. In only remains to verify that $a$ does not violate the IR constraint. This holds because $\int_{\und{\theta}}^\theta a(z) dz \leq \int_{\und{\theta}}^\theta \tilde{a}(z) dz$ for all $\theta \in [\und{\theta},\bar{\theta}]$
\end{proof}

\subsection{Proof of \texorpdfstring{\Cref{thm:convex}}{}}\label{proof:convex}

\begin{proof}
    We can write the objective from \cref{eq:theta_prime} more explicitly as a function of $\theta'$ as
    \[
    \frac{R(\theta') + (\bar{\theta}^{\#} - \theta')R'(\theta')}{\bar{\theta}^{\#}} \mathcal{W}(\bar{\theta}^{\#}) + R'(\theta')\left( \mathcal{W}(\und{\theta}^*) - \mathcal{W}(\bar{\theta}^{\#}) \right) + w^B\und{\theta} \frac{R(\theta') + (\bar{\theta}^{\#} - \theta')R'(\theta')}{\bar{\theta}^{\#}}.
    \]
Suppose $R$ is twice differentiable (if not we can use the left or right derivative of $R'$ and apply the same argument). Then the derivative with respect to $\theta'$ is
\begin{align*}
    &\frac{R''(\theta')(\bar{\theta}^{\#} - \theta')}{\bar{\theta}^{\#}} \mathcal{W}(\bar{\theta}^{\#}) + R''(\theta')\left( \mathcal{W}(\und{\theta}^*) - \mathcal{W}(\bar{\theta}^{\#}) \right) + w^B\und{\theta} \frac{R''(\theta')(\bar{\theta}^{\#} - \theta')}{\bar{\theta}^{\#}} \\ 
    &= R''(\theta')\left( \frac{(\bar{\theta}^{\#} - \theta')}{\bar{\theta}^{\#}} \mathcal{W}(\bar{\theta}^{\#}) + \left( \mathcal{W}(\und{\theta}^*) - \mathcal{W}(\bar{\theta}^{\#}) \right) + w^B\und{\theta} \frac{(\bar{\theta}^{\#} - \theta')}{\bar{\theta}^{\#}} \right) 
\end{align*}
Thus $\theta'$ satisfies the FOC iff 
\[
\frac{(\bar{\theta}^{\#} - \theta')}{\bar{\theta}^{\#}} \mathcal{W}(\bar{\theta}^{\#}) + \left( \mathcal{W}(\und{\theta}^*) - \mathcal{W}(\bar{\theta}^{\#}) \right) + w^B\und{\theta} \frac{(\bar{\theta}^{\#} - \theta')}{\bar{\theta}^{\#}} = 0
\]
which simplifies to 
\[
\theta' = \bar{\theta}^{\#} \frac{\mathcal{W}(\und{\theta}^*) + w^B\und{\theta}}{\mathcal{W}(\bar{\theta}^{\#}) + w^B\und{\theta}}.
\]
We then need to verify that $\theta' \leq \hat{\theta}$. This leads to the stated solution.
\end{proof}

\subsection{Proof of \texorpdfstring{\Cref{prop:general_hatR}}{}}\label{proof:general_hatR}

\begin{proof}
    Property 1 is immediate, since for any concave $\hat{R}$ the solution will set $\alpha = 0$ above $\und{\theta}^*$. Property 2 follows from the separating hyperplane theorem. Specifically, if $U \leq R$ on $[\und{\theta},\bar{\theta}]$ for some concave $U$, then there is a hyperplane separating $\{(x,y) : x \in X_i^{cv}, \ y \geq R(x)\}$ from the hypograph of $U$. This hyperplane defines the slope $s_i$.

    Property 3 follows from property 2. To see this, notice that if $s_i > s_j$ for $i > j$ then either $\hat{R} < R$ over either $X_j^{cv}$ or $X_i^{cv}$. If this holds for $j$ then we can increase $\hat{R}$ pointwise by increasing $s_j$. If it holds for $i$ then we can increase $\hat{R}$ pointwise by reducing $s_j$. 

    For property 4, note that from \Cref{prop:concave_relax} we know that we can restrict attention to $(\alpha,\und{u})$ such that $\alpha(\theta) \leq \min\{R(\und{\theta})/\und{\theta}), R(\bar{\theta}^{\#})/\bar{\theta}^{\#} \}$ for all $\theta$. Under this restriction, increasing $s_i$ above $\min\{R(\und{\theta})/\und{\theta}), R(\bar{\theta}^{\#})/\bar{\theta}^{\#} \}$ only shrinks the feasible set.   
\end{proof}

\subsection{Proof of \texorpdfstring{\Cref{thm:outer}}{}}\label{proof:outer}

\begin{proof}
We want to establish strong duality. If the program is feasible then there is an aggregate allocation $(\hat{n}_j^A,\hat{n}_j^B)_{j \in \mathcal{J}}$ which satisfies market clearing and $(\hat{n}_j^A,\hat{n}_j^B) \in \mathcal{F}_j$ for all $j$.  If, moreover, $(\hat{n}_j^A,\hat{n}_j^B)$ is in the interior of $\mathcal{F}_j$ for all $j$ then Slater's condition is satisfied, so we are done. Otherwise, define a perturbed problem by letting
\[
\mathcal{F}_j^{\varepsilon} = \{(\hat{n}^h,\hat{n}^l) \in \R^2 \ : \exists \ (x,y) \in \mathcal{F}_j \text{ with } |(x,y) - (\hat{n}^h,\hat{n}^l)| \leq \varepsilon\}.
\]
for $\varepsilon \geq 0$, and let $S_j^{\varepsilon}$ be the support function of this set. Then $(\hat{n}_j^A,\hat{n}_j^B)$ is in the interior of $\mathcal{F}_j^{\varepsilon}$ for all $j$, so Slater's condition is satisfied for any $\varepsilon > 0$ if we replace $\mathcal{F}_j$ with $\mathcal{F}_j^{\varepsilon}$. Let $P_{\varepsilon}$ and $D_{\varepsilon}$ be the values of the primal and dual programs at $\varepsilon$. Notice that $\mathcal{F}_j^{\varepsilon}$ is compact and convex valued and $\varepsilon \mapsto \mathcal{F}_j^{\varepsilon}$ is upper- and lower-hemicontinuous on $\R_+$. Thus the value of the primal outer program is continuous in $\varepsilon$ by Berge's maximum theorem, and converges to $P_0$ as $\varepsilon \rightarrow 0$.

Now note that $\varepsilon \mapsto S_j^{\varepsilon}(w^A,w^B)$ is increasing, so $\varepsilon \mapsto D_{\varepsilon}$ is decreasing. Moreover, weak duality (for a minimization problem) implies that $P_{\varepsilon} \geq D_{\varepsilon}$ for all $\varepsilon \geq 0$. Then $P_{\varepsilon} = D_{\varepsilon} \leq D_{0} \leq P_0$ for all $\varepsilon$ and $P_\varepsilon \rightarrow P_0$ as $\varepsilon \rightarrow 0$ imply $D_0 = P_0$.  
\end{proof}

\subsection{Proof of \texorpdfstring{\Cref{prop:sub_dispersion}}{}}\label{proof:sub_dispersion}

\begin{proof}
    The equivalence between $i.$ and $ii.$ is straightforward: the envelop theorem \citep{milgrom2002envelope} implies that $ S_j(x',y) - S_j(x,y) = \int_{x}^{x'} N_j^A(z,y) dz$, which is decreasing in $y$ for any $x' > x$ if and only if $y \mapsto N_j^A(z,y)$ is decreasing for all $z$.

    We now show the equivalence of $i.$ and $iii.$. Suppress the $j$ subscript to simplify notation. Regardless of whether or not $R$ is concave, we can replace it with an effective upper bound $\hat{R}$ which is. Moreover, the assumptions that $R'(\und{\theta}) > 0$ and $R(\bar{\theta}) > R(\und{\theta})$ imply that $\hat{R}'(\und{\theta}) > 0$ (see \Cref{thm:convex} and \Cref{prop:general_hatR}). Thus we can proceed assuming that $R$ is concave and $R'(\und{\theta}) > 0$.
    
    Because $\mathcal{F}$ is convex, it is downwards closed if and only if $0 \in N^A(0,1)$ and $0 \in N^B(1,0)$. Consider first $N^A(0,1)$. Then $\mathcal{W}(\theta) + \und{\theta}w^B = \theta(1 - F(\theta))$ which is strictly positive for all $\theta < \bar{\theta}$. Because $R'(\und{\theta}) > 0$, \Cref{prop:concave_relax} implies that $0 \in N^A(0,1)$ if and only if $\und{\theta}^* = \und{\theta}$, or equivalently 
    \[
    \bar{W}(\theta) \leq 0 \ \ \forall \ \theta \in [\und{\theta},\bar{\theta}].
    \]
    This holds if and only if
    \[
    \mathcal{W}(\theta) \leq 0 \ \ \forall \ \theta \in [\und{\theta},\bar{\theta}].
    \]
    For $w^A = 0$ and $w^B = 1$ we have 
    \begin{align*}
    \mathcal{W}(\theta) = (\theta - \und{\theta}) - F(\theta)\theta. 
    \end{align*}
    Then we have $0 \in N^A(0,1)$ if and only if $(\theta - \und{\theta})/\theta \leq F(\theta)$ for all $\theta$, which gives us one half of the condition that $F$ has low dispersion. 

    Next, we consider whether $0 \in N^B(1,0)$. Here $\mathcal{W}(\theta) + \und{\theta}w^B = F(\theta) + \und{\theta} w^B$ which is strictly positive for all $\theta$. Then \Cref{prop:concave_relax} implies that $0 \in N^B(1,0)$ if and only if $\bar{\theta}^{\#} = \bar{\theta}$. By definition of $\bar{\theta}^{\#}$, this holds if and only if 
    \[
    \bar{\theta} \in \argmax \frac{F(\theta)}{\theta},
    \]
    or equivalently $\theta/\bar{\theta} \geq F(\theta)$ for all $\theta$, which gives us the second half of the condition that $F$ has low dispersion.
\end{proof}

\subsection{Proof of \texorpdfstring{\Cref{thm:no_disposal}}{}}\label{proof:no_disposal}

We first prove a preliminary result (it is only used for the converse direction).
Define $\mathcal{M}(x): \mathcal{F} \mapsto \mathcal{F}^J$ be the set of incentive-feasible market-clearing allocation profiles given aggregate labor allocation $\ell = xJ$. Equivalently, $\mathcal{M}(x)$ is the set of incentive-feasible allocation profiles with barycenter $x$.
    \begin{lemma}\label{lem:uhc}
        $\mathcal{M}$ is non-empty; compact and convex-valued; and upper-hemicontinuous on $\mathcal{F}$. Additionally, it is single valued, and hence lower-hemicontinuous, on the frontier of $\mathcal{F}$. 
    \end{lemma}
    \begin{proof}
        $\mathcal{M}$ is non-empty since $(x,\dots,x) \in \mathcal{M}(x)$ for all $x \in \mathcal{F}$. It is compact and convex valued because $\mathcal{M}(x)$ is the intersection of $\mathcal{F}^J$ with the two hyperplanes defined by the market-clearing conditions. For upper-hemicontinuity, define $A(y) := \frac{1}{J} \sum_{j=1}^J y_j$ for $y \in \mathcal{F}^J$. The graph of $\mathcal{M}$ is given by $
        (\R^2 \times \mathcal{F}^J) \cap \{ (x,y) : A(y) = x \}$ which is closed because $A$ is continuous. Because $\mathcal{M}$ is compact-valued, it is upper-hemicontinuous. $M$ is single-valued on the frontier because the frontier is strictly convex, so $(x,\dots,x)$ is the only incentive-feasible profile with barycenter $x$. Upper-hemicontinuous correspondences are lower-hemicontinuous where they are single-valued. 
    \end{proof}

\noindent\textit{Proof of \Cref{thm:no_disposal}}
\textit{Part i.} The objective in the outer program
    \[
    \sup_{\lambda^A,\lambda^B}  \ell^A \lambda^A  + \ell^B \lambda^B  - \sum_{j\in \mathcal{J}} S_j\left(  \lambda^A - \pi_j^A  \ , \ \lambda^B - \pi_j^B \right) 
    \]
    has increasing differences in $(\ell^A,\ell^B)$ and $\lambda^k$ for either $k$. Moreover, by \Cref{prop:sub_dispersion} the objective is supermodular in $(\lambda^A,\lambda^B)$ if heterogeneity is small. This implies $\ell^{k} \mapsto \lambda_*^{k}$ and $\ell^{k} \mapsto \lambda_*^{-k}$ are increasing \citep{milgrom1994monotone}. In particular, this means negative price spillovers for any increase in automation. 

    \textit{Part ii.} Let $\ell' \leq \ell$. We have
    \begin{align*}
    \left(\lambda_*^A(\ell) - \pi_j^A \right)n_j^A(\ell) + \left(\lambda_*^B(\ell) - \pi_j^B \right)n_j^B(\ell) &= \max_{(n^A,n^B) \in \mathcal{F}_j} \left(\lambda_*^A(\ell) - \pi_j^A \right)n^A + \left(\lambda_*^B(\ell) - \pi_j^B \right)n^B \\
    & \geq \left(\lambda_*^A(\ell) - \pi_j^A \right)n_j^A(\ell') + \left(\lambda_*^B(\ell) - \pi_j^B \right)n_j^B(\ell') \\
    & \geq \left(\lambda_*^A(\ell') - \pi_j^A \right)n_j^A(\ell') + \left(\lambda_*^B(\ell') - \pi_j^B \right)n_j^B(\ell')
    \end{align*}
    where the equality in the first line is the definition of $(n^A_j(\ell), n^B_j(\ell))$, and the inequality in the third line holds because $(\lambda_*^A(\ell'), \lambda_*^B(\ell')) \leq (\lambda_*^A(\ell), \lambda_*^B(\ell))$, as shown in part $i$. 

    \textit{Part iii.} For any $j$ that is on the frontier of $\mathcal{F}_j$ under the original labor allocation, $\ell$, it cannot be that the output in both tasks increases, because \Cref{prop:sub_dispersion} implies that there are no such points in $\mathcal{F}_j$. If $j$ is off the frontier under $\ell$, it must be that $(\lambda_*^A(\ell) - \pi_j^A, \lambda_*^B(\ell) - \pi_j^B) \leq 0$. Then by part $i.$, for any $\ell' < \ell$ we have $(\lambda_*^A(\ell') - \pi_j^A, \lambda_*^B(\ell') - \pi_j^B) \leq 0$. By \Cref{cor:outer}, part $iii.$, there are two cases to consider. First, if the allocation of some group $j$ is in the interior of $\mathcal{F}_j$, then $j$ is the only group off the frontier. For this to be the case, it must be that $(\lambda_*^A(\ell) - \pi_j^A, \lambda_*^B(\ell) - \pi_j^B) = 0$. Then the only way $j$'s allocation can weakly increase is if $\lambda_*^A(\ell) = \lambda_*^A(\ell')$ and $ \lambda_*^B(\ell) = \lambda_*^B(\ell')$. But then the allocation of all agents on the frontier is unchanged, so $j$'s allocation cannot weakly increase. In the other case there are two groups off the frontier, one, $j$, with $\lambda^A = \pi^A_j$, $\lambda^B < \pi^B_j$ and $n_j^B(\ell) = 0$; and the other, $j'$, with $\lambda^B = \pi^B_{j'}$, $\lambda_A < \pi^A_{j'}$ and $n_{j'}^A(\ell) = 0$. Suppose $j$'s allocation weakly increases, which, because $(\lambda_*^A(\ell') - \pi_j^A, \lambda_*^B(\ell') - \pi_j^B) \leq 0$,  would mean $j$ receives more units of $A$ and still no units of $B$. This can happen only if $\lambda_*^A(\ell') = \lambda_*^A(\ell)$. But then by \Cref{prop:sub_dispersion} all other agents receive weakly more units of task $A$ as well, which cannot be. A symmetric argument applies to $j'$. 

    For the converse results (parts $iv.$ - $vi.$) the symmetry assumption implies $\mathcal{F}_j = \mathcal{F}_{j'}$ and $S_j = S_{j'}$ for all $j,j'$. I therefore suppress dependence of these objects on $j$ in the notation.

    By \Cref{prop:sub_dispersion}, the frontier of $\mathcal{F}$ has at least one upward-sloping portion, near one of the two axes. Without loss of generality, assume that the frontier is upward sloping for low values of $A$, as in \Cref{fig:high_heterogeneity}, and focus on this segment.  

    \textit{Claim 1.} There exists a point $\bar{x}$ on the upward-sloping portion of the frontier of $\mathcal{F}$, and a neighborhood $V$ of $x$ such that for all $x' \in V \cap \mathcal{F}$, the allocation $(n_j^A(J x'), n_j^B(Jx'))$ is on the upward-sloping portion of the frontier of $\mathcal{F}$ for every $j$. Moreover, $\lambda_*^A(\ell') \rightarrow - \infty$ and $\lambda_*^B(\ell') \rightarrow +\infty$ as $\ell \rightarrow J\bar{x}$. 

    To prove Claim 1, recall that $M(\bar{x}) = (\bar{x},\dots,\bar{x})$ for $\bar{x}$ on the frontier. Then by upper-hemicontinuity of $M$, for any open set $\bar{U} \subset \R^{2J}$ containing $(\bar{x},\dots,\bar{x})$ there exists a neighborhood $V$ of $\bar{x}$ such that $M(x') \subset \bar{U}$ for all $x' \in V$. Equivalently, there exists a neighborhood $U \subset \R^2$ containing $\bar{x}$ such all groups receive allocations in $U$ for any profile in $M(x')$. Take $U$ small enough such that it only intersects the boundary of $\mathcal{F}$ on the upward-sloping part of the frontier. Then for any $x' \in V$, \Cref{cor:outer} part $iii.$ implies that $(n_j^A(Jx'), n_j^B(Jx'))$ is on the upward-sloping part of the frontier for all but at most one $j$. 

    Because the frontier is convex, it is smooth almost everywhere. So we can choose $\bar{x}$ such that $\mathcal{F}$ has a unique supporting hyperplane at $\bar{x}$. Let $(\eta^A,\eta^B)  \in \R^2$ be the normal vector of this hyperplane. Because we are focusing on the upward-sloping portion of the frontier near the $B$-axis, $\eta^A < 0$ and $\eta^B > 0$.
    
    Because performance is distinct across groups and $J \geq 3$, there are at least two groups $j,j'$ on the frontier, and we have $\lambda_*^A(Jx') < \pi_j^A,\pi_{j'}^A$ and $\lambda_*^B(Jx') > \pi_j^B,\pi_{j'}^B$. For their allocations to be in $U$ as this set shrinks, it must be that $(\lambda_*^A(Jx') - \pi_j^A, \lambda_*^B(Jx') - \pi_j^B)$ and $(\lambda_*^A(Jx') - \pi_{j'}^A, \lambda_*^B(Jx') - \pi_{j'}^B)$ become parallel to $(\eta^A,\eta^B)$. Because $\pi_j^A \neq \pi_{j'}^A$ and $\pi_j^B \neq \pi_{j'}^B$, this can only happen if $\lambda_*^A \rightarrow - \infty$ and $\lambda_*^B \rightarrow +\infty$. 

    Putting this together, we have that $\lambda_*^A(Jx') \rightarrow - \infty$ and $\lambda_*^B(Jx') \rightarrow +\infty$ as $x' \rightarrow \bar{x}$. Then for $x'$ sufficiently close to $\bar{x}$ we have $(\lambda_*^A(Jx') - \pi_j^A < 0$ and $\lambda_*^B(Jx') - \pi_j^B) > 0$ for all $j$, so every group is on the frontier. This proves Claim 1. 

    We now show part $iv.$ of the result. Let $\bar{x} = (\bar{x}^A,\bar{x}^B)$ and $V$ be the point and neighborhood from Claim 1, and let $\hat{x} = (\bar{x}^A + \varepsilon,\bar{x}^B) \in V \cap \mathcal{F}$. Consider the increase in automation of reducing $\varepsilon$. Because all groups are on the upward-sloping part of the frontier for $x$ on the line segment $(\bar{x},\hat{x})$  
    \[
     \ell^A \lambda^A  + \ell^B \lambda^B  - \sum_{j\in \mathcal{J}} S_j\left(  \lambda^A - \pi_j^A  \ , \ \lambda^B - \pi_j^B \right) 
    \]
    is submodular in $\lambda^A,\lambda^B$ on a neighborhood of $(\lambda_*^A(J\hat{x}), \lambda_*^B(J\hat{x}))$. So $\lambda_*^A$ decreases and $\lambda_*^B$ increases as $\varepsilon \rightarrow 0$, i.e., there are positive price spillovers. 

    For part $v.$ consider the same increase in automation. As shown in the proof of Claim 1, $(\lambda^A_*(Jx') / \lambda_*^B(Jx')) \rightarrow \eta^A/\eta^B$ and $|(\lambda^A_*(Jx') , \lambda_*^B(Jx'))| \rightarrow +\infty$ as $x' \rightarrow \bar{x}$. Moreover $\eta \cdot \bar{x} > 0$. Because each group's allocation converges to $\bar{x}$ as $\varepsilon \rightarrow 0$, the surplus generated by each group goes to $+\infty$. 

    For part $vi.$, take any $x = (x^A,x^B) \in \mathcal{F}\cap V$ such that $x^A < \bar{x}^A$ and $x^B < \bar{x}^B$. As $(x^A,x^B) \in V$, all groups allocations are on the frontier. Moreover, they cannot all be strictly greater than $\bar{x}$ without violating market clearing. Thus some are strictly below $x$. Since all allocations converge to $\bar{x}$ as $x \rightarrow \bar{x}$, the allocations of some groups must strictly increase.  \hfill\qedsymbol

\subsection{Proof of \texorpdfstring{\Cref{prop:eqm_set}}{}}\label{proof:eqm_set}
\begin{proof}
    By the envelope theorem, the right derivative of $L$ with respect to $\ell^k$ is $\max \lambda_*^K(\ell^A,\ell^B)$. Then for any $\ell$ such that $\ell/J \in \mathcal{F} \setminus \mathcal{E}$, the labor-cost is weakly decreasing in at least one dimension. Because $Y$ is strictly increasing, it is possible to produce the same output at strictly lower labor cost. Hence $\ell$ cannot be optimal. Conversely, for any $\ell$ such that $\ell/J \in \mathcal{E}$ we can define $P,Y,$ and $\gamma$ such that the firm's profit is concave in $\ell$ and the iso-profit curve is tangent to the iso-labor-cost curve at $\ell$. 
\end{proof}

\subsection{Proof of \texorpdfstring{\Cref{prop:eqm_properties}}{}}\label{proof:eqm_properties}

\begin{proof}
    We first show that $\mathcal{E} = \mathcal{F}$ under low heterogeneity. That is, we need to show that $\max\lambda_*^A(J n^A, J n^B), \max\lambda_*^A(J n^A, J n^B) > 0$ for all $(n^A,n^B) \in \mathcal{F}$. Recall that the outer program defining $\lambda_*$ is 
    \[
    \sup_{\lambda^A,\lambda^B}   \ell^A \lambda^A  + \ell^B \lambda^B  - \sum_{j\in \mathcal{J}} S_j\left(  \lambda^A - \pi_j^A  \ , \ \lambda^B - \pi_j^B \right).
    \]
    Observe that $(\lambda^A,\lambda^B) = (\min_j \pi_j^A, \min_j \pi_j^B)$ is a solution to this program for $(\ell^A,\ell^B) = (0,0)$, because it achieves a value of $0$ and $S_j \geq 0$. By \Cref{prop:sub_dispersion}, under low heterogeneity the objective in the outer program is supermodular. Then because $(n^A,n^B) \geq 0$ for all $(n^A,n^B) \in \mathcal{F}$, we have $\max \lambda_*^A(J n^A, J n^B) \geq \min_j \pi_j^A > 0$ and $\max \lambda_*^B(J n^A, J n^B) \geq \min_j \pi_j^B > 0$ for all $n^A,n^B \in \mathcal{F}$ (\citealt{milgrom1994monotone}, Theorem 4).

    That the equilibrium set must exclude the points on the upward-sloping region of the frontier where it is smooth was established in the proof of \Cref{thm:no_disposal}, where we showed that either $\lambda^A_*$ or $\lambda^B_* \rightarrow - \infty$ as $(\ell^A/J, \ell^B/J)$ approaches such a point. Because \Cref{prop:sub_dispersion} shows that the frontier of $\mathcal{F}$ has upward-sloping segments if and only if preference heterogeneity is high, this also shows that low heterogeneity is necessary for $\mathcal{E} = \mathcal{F}$. Similarly, as $(\ell^A/J, \ell^B/J)$ approaches any point on the downward-sloping region of the frontier, it must be that $\lambda^A_* - \pi_j^A, \lambda^B_* - \pi_j^A \geq 0$ for all $j$, so such points are in $\mathcal{E}$. 

    Part \textit{iii.} follows from the fact, previously observed, that the outer program has increasing differences in $\ell^k$ and $\lambda^k$. 
\end{proof}

\subsection{Proof of \texorpdfstring{\Cref{prop:specialization}}{}}\label{proof:specialization}
\begin{proof}
    Let $\lambda^A_t,\lambda^B_t$ be the solution to the dual outer program in period $t$. Consider the inner program for group $j$ in period $t$ given weights $(\lambda^A_t - \pi_j^t, \lambda^B_t - \pi_j^t) \not\leq (0,0)$. Because $R_j^{t}$ is the indirect utility from the mechanism in period $t - 1$, we have $R_j^t(\und{\theta})/\theta = a_j^{t-1}(\theta) + b_j^{t-1}(\theta)/\theta \geq a^{t-1}(\theta) = R_j'(\theta)$. Thus $\check{\theta} = \und{\theta}$. Let $\mathcal{W}_j^t$ be the function defined as in \cref{eq:mathcalW} using weights $(\lambda^A_t - \pi_j^t, \lambda^B_t - \pi_j^t)$. Then by \Cref{prop:concave_relax}, any solution to the inner program is defined by some
        $\tilde{\theta} \in [\und{\theta}^{\#}, \bar{\theta}^{\#}]$ such that $\und{u} = \und{\theta} R(\tilde{\theta})/\tilde{\theta}$ and 
    \begin{equation}\label{eq:a_t}
        a^t(\theta) = 
        \begin{cases}
        \frac{\und{u}}{\und{\theta}} \quad &\text{on } [\und{\theta},  \min\{\tilde{\theta}, \und{\theta}^*\}] \\
        R'(\theta) &\text{on } (\tilde{\theta}, \und{\theta}^*] \\
        \text{non-increasing and bounded above by } R'(\und{\theta}^*) & \text{on } (\und{\theta}^*,\bar{\theta}^*] \\
        0 &\text{on } (\bar{\theta}^*, \bar{\theta}]
        \end{cases}
    \end{equation}
    
    For the frontier of $\mathcal{F}_j^t$ to be strictly convex, it must be that there is a unique solution to the inner program for any $(w^A,w^B) \not\leq (0,0)$. This implies that $\und{\theta}^{\#} = \bar{\theta}^{\#}$ and $\und{\theta}^* = \bar{\theta}^*$. This in turn implies $\bar{\theta}^{\#}$ and $\bar{\theta}^*$ do not belong to ironing intervals, i.e., $\mathcal{W}_j^t(\bar{\theta}^{\#}) = \bar{\mathcal{W}}_j^t(\bar{\theta}^{\#})$ and $\mathcal{W}_j^t(\bar{\theta}^*) = \bar{\mathcal{W}}_j^t(\bar{\theta}^*)$. In fact, this equality must hold for any $\theta \leq \bar{\theta}^{\#}$, since otherwise there would be some $(w^A,w^B) \not\leq (0,0)$ for which there are multiple solutions.
    
    Define $\bar{\theta}^{\#} = \theta^{\#}_t$ and $\bar{\theta}^* = \theta^*_t$ to make the dependence on the period-$t$ parameters explicit. Then the indirect utility $R^t_j(\theta)$ satisfies 
    \[
    R^t_j(\theta) = 
    \begin{cases}
        \theta \frac{R_j^{t-1}(\bar{\theta}^{\#})}{\bar{\theta}^{\#}} &\text{if } \theta \leq \bar{\theta}^{\#} \\
        R_j^{t-1}(\bar{\theta}^*) & \text{if } \theta \geq \bar{\theta}^* \\
        \text{increasing from $R_j^{t-1}(\bar{\theta}^{\#})$ to $R_j^{t-1}(\bar{\theta}^*)$} & \text{on } [\bar{\theta}^{\#}, \bar{\theta}^*].
    \end{cases}
    \]

    Now consider a machine-enhancing change from period $t$ to period $t+1$. This shifts $\lambda^A,\lambda^B$ and thus changes the weights $(\lambda^A - \pi_j^A, \lambda^B - \pi_j^B)$ defining $\mathcal{W}^{t+1}_j$. Let $\tilde{\theta}_{t+1}$ be the new threshold defining $a^{t+1}$, as in \cref{eq:a_t}. Notice that because of the shape of $R_j^t$, if $\tilde{\theta}_{t+1} \leq \theta^{\#}_t$ then the mechanism is unchanged for $\theta \leq \theta^{\#}_t$. Similarly, if $\und{\theta}^*_{t+1} \geq \theta^*_t$ then the mechanism is unchanged for $\theta \geq \theta^*_t$. Then the only possible changes come from increasing $\tilde{\theta}_{t+1}$ above $\theta^{\#}_t$, or reducing $\und{\theta}^*_{t+1}$ below $\theta^*_t$. In either case we get an increase in specialization. Moreover, the increase in specialization is uniform if $\mathcal{W}_j^{t+1} = \bar{\mathcal{W}}_j^{t+1}$ on $[\tilde{\theta}_{t+1}, \und{\theta}^*_{t+1}$, which must be the case if the frontier of $\mathcal{F}_j^{t+1}$ is strictly convex (by the same argument applied above to $\mathcal{F}_j^{t}$). 
\end{proof}

\subsection{Proof of \texorpdfstring{\Cref{thm:competitive}}{}}\label{proof:competitive}    

\begin{proof}
    The strategy is to first describe the solution to the monopsony program, assuming full participation. Because this program drops the constraint (C) it is clearly a lower bound on the production cost in the competitive market. We then set the prices in the competitive market equal to the multipliers from the monopsony program, and find wages such that it is optimal for each firm to induce full participation. Finally, we verify that under full participation, $(a_j^*,b_j^*)$ is an optimal contract.

    With weak inequalities on the market-clearing constraints in \cref{prog:1}, the multipliers on these constraints must be non-negative. Solving this program is thus equivalent to replacing the labor-cost function with 
    \[
    \tilde{L}(\ell^A,\ell^B) := \sup_{\lambda^A,\lambda^B \geq 0} \left\{\ell^A \lambda^A + \ell^B\lambda^B - \sum_{j \in \mathcal{J}} S_j\left(\lambda^A - \pi_j^A, \lambda^B - \pi_j^B\right) \right\}. 
    \]
    In other words, we only consider labor allocations in $\tilde{\mathcal{E}}$. Otherwise, the monopolist's program is unchanged. In particular, the inner program is identical. Let $\lambda_*^A,\lambda_*^B$ be the multipliers and $a_m^*,b_m^*, (a_j^*,b_j^*)$ the machine allocations and mechanisms defining this solution (note that these mechanisms are defined in terms of the ratio $\theta = \theta^A/\theta^B$, but I abuse notation an write them also as a function of the pair $(\theta^A,\theta^B)$). 

    Setting the wage for group $j$ at 
    \[
    \omega_j = \min_{\theta \in [\und{\theta}_j, \bar{\theta}_j]} \left(\lambda_*^A - \pi_j^A\right) a_j^*(\theta) +  \left(\lambda_*^B - \pi_j^B\right) b_j^*(\theta).
    \]
    and consider the candidate equilibrium in which $x_j^* = 1$ for all $j$, along with the monopsony machine allocations and mechanisms. 
    
    To show that this satisfies the best-response condition for firms, we first show that it is optimal for firms to induce full participation. Suppose that it is optimal for some firm to offer a contract $(x_j,a_j,b_j)$ for group $j$ that does not induce full participation. Consider the modification which instead offers $(a_j^*(\theta),b_j^*(\theta))$ to any type such that $x_j(\theta^A,\theta^B) = 0$ and $\theta = \theta^A/\theta^B$. This modification increases profits given the specification of $\omega_j$. It does not disturb the (IR) constraint because (IR) is imposed in the monopsony program and is thus satisfied by $a^*_j,b^*_j$. It does not disturb the (IC) constraint because if $x_j(\theta^A,\theta^B) = 1$ then
    \begin{align*}
    \theta^A a_j(\theta^A,\theta^B) + \theta^B b_j(\theta^A,\theta^B) &\leq  \theta^A a^*_j(\theta^A,\theta^B) + \theta^B b^*_j(\theta^A,\theta^B) \\
    &\leq \theta^A a^*_j(\hat{\theta}^A,\hat{\theta}^B) + \theta^B b^*_j(\hat{\theta}^A,\hat{\theta}^B)
    \end{align*}
    for any $(\hat{\theta}^A,\hat{\theta}^B)$, where the first inequality holds because $(x_j,a_j,b_j)$ satisfies (C), and the second because $(a_j^*,b_j^*)$ satisfies (IC). Finally, this modification does not violate (C) if $U \leq \omega_j - \theta^A a^*_j(\theta^A,\theta^B) - \theta^B b^*_j(\theta^A,\theta^B)$ for all $(\theta^A,\theta^B) \in \Theta_j$, i.e., if $U$ is sufficiently low. Under this condition, we include that it is optimal for firms to induce full participation. 

    Finally, we need to show that under full participation, $(a^*_j,b^*_j)$ is a solution to the firm's program when it is also the mechanism defining the participation constraint. To see this, recall (e.g., from \cref{prog:outer_dual2}) that $(a^*_j,b^*_j)$ is a solution to the inner program
    \begin{align}
    &\max_{a,b\geq 0}  \int\left( (\lambda_*^A - \pi_j^A) a(\theta^A,\theta^B) + (\lambda_*^B - \pi_j^B) b(\theta^A,\theta^B) \right) dF_j(\theta^A,\theta^B) \notag\\
         s.t.& \notag \\
         & \theta^A a(\theta^A,\theta^B) + \theta^B b(\theta^A,\theta^B) \leq  \theta^A a(\hat{\theta}^A,\hat{\theta}^B) + \theta^B b(\hat\theta^A,\hat\theta^B) \quad  \forall \ (\theta^A,\theta^B), (\hat\theta^A,\hat\theta^B) \in \Theta_j \tag{IC} \\
       & \theta^A a(\theta^A,\theta^B) + \theta^B b(\theta^A,\theta^B) \leq r_j(\theta^A,\theta^B) \quad \quad \forall \ (\theta^A,\theta^B) \in \Theta_j\tag{IR}
\end{align}
Under full participation, this is just a relaxation of the firm's program in the competitive case, without the (C) constraint. Since (C) is satisfied if all other firms use the same mechanism, $(a^*_j,b^*_j)$ is also optimal in the competitive setting.  
\end{proof}

\begin{singlespace}
\bibliography{references}

\begin{thebibliography}{36}
\providecommand{\natexlab}[1]{#1}
\providecommand{\url}[1]{\texttt{#1}}
\expandafter\ifx\csname urlstyle\endcsname\relax
  \providecommand{\doi}[1]{doi: #1}\else
  \providecommand{\doi}{doi: \begingroup \urlstyle{rm}\Url}\fi

\bibitem[Acemoglu(2025)]{acemoglu2025simple}
D.~Acemoglu.
\newblock The simple macroeconomics of ai.
\newblock \emph{Economic Policy}, 40\penalty0 (121):\penalty0 13--58, 2025.

\bibitem[Acemoglu and Restrepo(2018)]{acemoglu2018race}
D.~Acemoglu and P.~Restrepo.
\newblock The race between man and machine: Implications of technology for
  growth, factor shares, and employment.
\newblock \emph{American economic review}, 108\penalty0 (6):\penalty0
  1488--1542, 2018.

\bibitem[Acemoglu and Restrepo(2024)]{acemoglu2024automation}
D.~Acemoglu and P.~Restrepo.
\newblock Automation and rent dissipation: Implications for wages, inequality,
  and productivity.
\newblock Technical report, National Bureau of Economic Research, 2024.

\bibitem[Acemoglu and Zilibotti(2001)]{acemoglu2001productivity}
D.~Acemoglu and F.~Zilibotti.
\newblock Productivity differences.
\newblock \emph{The Quarterly Journal of Economics}, 116\penalty0 (2):\penalty0
  563--606, 2001.

\bibitem[Acemoglu et~al.(2024)Acemoglu, Kong, and Restrepo]{acemoglu2024tasks}
D.~Acemoglu, F.~Kong, and P.~Restrepo.
\newblock Tasks at work: comparative advantage, technology and labor demand.
\newblock Technical report, National Bureau of Economic Research, 2024.

\bibitem[Baron et~al.(2025)Baron, Lombardo, Ryan, Suh, and
  Valenzuela-Stookey]{baron2025mechanism}
E.~J. Baron, R.~Lombardo, J.~P. Ryan, J.~Suh, and Q.~Valenzuela-Stookey.
\newblock Mechanism reform for task allocation.
\newblock Technical report, National Bureau of Economic Research, 2025.

\bibitem[Border(1991)]{border1991implementation}
K.~C. Border.
\newblock Implementation of reduced form auctions: A geometric approach.
\newblock \emph{Econometrica}, pages 1175--1187, 1991.

\bibitem[Che et~al.(2013)Che, Kim, and Mierendorff]{che2013generalized}
Y.-K. Che, J.~Kim, and K.~Mierendorff.
\newblock Generalized reduced-form auctions: A network-flow approach.
\newblock \emph{Econometrica}, 81\penalty0 (6):\penalty0 2487--2520, 2013.

\bibitem[{Destruction of Stocking Frames Act}(1812)]{stockingframes1812}
{Destruction of Stocking Frames Act}.
\newblock Destruction of stocking frames, etc. act 1812.
\newblock 52 Geo.~III, c.~16 (U.K.), 1812.
\newblock Statute of the Parliament of the United Kingdom.

\bibitem[Dinerstein and Smith(2021)]{dinerstein2021}
M.~Dinerstein and T.~D. Smith.
\newblock Quantifying the supply response of private schools to public
  policies.
\newblock \emph{American Economic Review}, 111\penalty0 (10):\penalty0
  3376--3417, 2021.

\bibitem[Dworczak and Muir(2024)]{dworczak2024mechanism}
P.~Dworczak and E.~Muir.
\newblock A mechanism-design approach to property rights.
\newblock {Working Paper}, 2024.

\bibitem[Dworczak et~al.(2021)Dworczak, Kominers, and
  Akbarpour]{dworczak2021redistribution}
P.~Dworczak, S.~D. Kominers, and M.~Akbarpour.
\newblock Redistribution through markets.
\newblock \emph{Econometrica}, 89\penalty0 (4):\penalty0 1665--1698, 2021.

\bibitem[Edmonds(1970)]{Edmonds1970}
J.~Edmonds.
\newblock Submodular functions, matroids, and certain polyhedra.
\newblock In R.~Guy, H.~Hanani, N.~Sauer, and J.~Sch{\"o}nheim, editors,
  \emph{Combinatorial Structures and Their Applications}, pages 69--87. Gordon
  and Breach, New York, 1970.
\newblock Proceedings of the Calgary International Conference on Combinatorial
  Structures and Their Applications, University of Calgary, June~1969.

\bibitem[Ellison et~al.(2004)Ellison, Fudenberg, and
  M{\"o}bius]{ellison2004competing}
G.~Ellison, D.~Fudenberg, and M.~M{\"o}bius.
\newblock Competing auctions.
\newblock \emph{Journal of the European Economic Association}, 2\penalty0
  (1):\penalty0 30--66, 2004.

\bibitem[Hart and Reny(2015)]{hart2015maximal}
S.~Hart and P.~J. Reny.
\newblock Maximal revenue with multiple goods: Nonmonotonicity and other
  observations.
\newblock \emph{Theoretical Economics}, 10\penalty0 (3):\penalty0 893--922,
  2015.

\bibitem[Hylland and Zeckhauser(1979)]{hylland1979efficient}
A.~Hylland and R.~Zeckhauser.
\newblock The efficient allocation of individuals to positions.
\newblock \emph{Journal of Political economy}, 87\penalty0 (2):\penalty0
  293--314, 1979.

\bibitem[Jullien(2000)]{jullien2000participation}
B.~Jullien.
\newblock Participation constraints in adverse selection models.
\newblock \emph{Journal of Economic Theory}, 93\penalty0 (1):\penalty0 1--47,
  2000.

\bibitem[Keynes(1930)]{keynes1930economic}
J.~M. Keynes.
\newblock Economic possibilities for our grandchildren.
\newblock In \emph{Essays in persuasion}, pages 321--332. Springer, 1930.

\bibitem[Kleiner et~al.(2021)Kleiner, Moldovanu, and
  Strack]{kleiner2021extreme}
A.~Kleiner, B.~Moldovanu, and P.~Strack.
\newblock Extreme points and majorization: Economic applications.
\newblock \emph{Econometrica}, 89\penalty0 (4):\penalty0 1557--1593, 2021.

\bibitem[Korinek and Suh(2024)]{korinek2024scenarios}
A.~Korinek and D.~Suh.
\newblock Scenarios for the transition to agi.
\newblock Technical report, National Bureau of Economic Research, 2024.

\bibitem[Lahr and Niemeyer(2024)]{lahr2024extreme}
P.~Lahr and A.~Niemeyer.
\newblock Extreme points in multi-dimensional screening.
\newblock \emph{arXiv preprint arXiv:2412.00649}, 2024.

\bibitem[Lewis and Sappington(1989)]{lewis1989countervailing}
T.~R. Lewis and D.~E. Sappington.
\newblock Countervailing incentives in agency problems.
\newblock \emph{Journal of economic theory}, 49\penalty0 (2):\penalty0
  294--313, 1989.

\bibitem[Maggi and Rodriguez-Clare(1995)]{maggi1995countervailing}
G.~Maggi and A.~Rodriguez-Clare.
\newblock On countervailing incentives.
\newblock \emph{Journal of Economic Theory}, 66\penalty0 (1):\penalty0
  238--263, 1995.

\bibitem[Martimort(1996)]{martimort1996exclusive}
D.~Martimort.
\newblock Exclusive dealing, common agency, and multiprincipals incentive
  theory.
\newblock \emph{The RAND journal of economics}, pages 1--31, 1996.

\bibitem[Milgrom and Segal(2002)]{milgrom2002envelope}
P.~Milgrom and I.~Segal.
\newblock Envelope theorems for arbitrary choice sets.
\newblock \emph{Econometrica}, 70\penalty0 (2):\penalty0 583--601, 2002.

\bibitem[Milgrom and Shannon(1994)]{milgrom1994monotone}
P.~Milgrom and C.~Shannon.
\newblock Monotone comparative statics.
\newblock \emph{Econometrica: Journal of the Econometric Society}, pages
  157--180, 1994.

\bibitem[Mussa and Rosen(1978)]{mussa1978monopoly}
M.~Mussa and S.~Rosen.
\newblock Monopoly and product quality.
\newblock \emph{Journal of Economic theory}, 18\penalty0 (2):\penalty0
  301--317, 1978.

\bibitem[Myerson(1981)]{myerson1981optimal}
R.~B. Myerson.
\newblock Optimal auction design.
\newblock \emph{Mathematics of operations research}, 6\penalty0 (1):\penalty0
  58--73, 1981.

\bibitem[{OpenAI}(n.d.)]{OpenAIOurStructure}
{OpenAI}.
\newblock Our structure.
\newblock \url{https://openai.com/our-structure/}, n.d.
\newblock Accessed 2025-10-09.

\bibitem[Phelps(2002)]{phelps2002lectures}
R.~R. Phelps.
\newblock \emph{Lectures on Choquet’s theorem}.
\newblock Springer, 2002.

\bibitem[Rochet(1987)]{rochet1987necessary}
J.-C. Rochet.
\newblock A necessary and sufficient condition for rationalizability in a
  quasi-linear context.
\newblock \emph{Journal of mathematical Economics}, 16\penalty0 (2):\penalty0
  191--200, 1987.

\bibitem[Roy(1951)]{roy1951some}
A.~D. Roy.
\newblock Some thoughts on the distribution of earnings.
\newblock \emph{Oxford economic papers}, 3\penalty0 (2):\penalty0 135--146,
  1951.

\bibitem[Ryff(1967)]{ryff1967extreme}
J.~V. Ryff.
\newblock Extreme points of some convex subsets of l1 (0, 1).
\newblock \emph{Proceedings of the American Mathematical Society}, pages
  1026--1034, 1967.

\bibitem[Schrijver(2003)]{schrijver2003combinatorial}
A.~Schrijver.
\newblock \emph{Combinatorial optimization: polyhedra and efficiency},
  volume~24.
\newblock Springer Science \& Business Media, 2003.

\bibitem[Stiglitz(1975)]{stiglitz1975theory}
J.~E. Stiglitz.
\newblock The theory of" screening," education, and the distribution of income.
\newblock \emph{The American economic review}, 65\penalty0 (3):\penalty0
  283--300, 1975.

\bibitem[Zeira(1998)]{zeira1998workers}
J.~Zeira.
\newblock Workers, machines, and economic growth.
\newblock \emph{The Quarterly Journal of Economics}, 113\penalty0 (4):\penalty0
  1091--1117, 1998.

\end{thebibliography}
\end{singlespace}

\end{document}